\documentclass[11pt,letterpaper]{article}

\usepackage[margin=1in]{geometry}
\usepackage{libertine}
\usepackage[libertine]{newtxmath}
\usepackage{microtype}
\usepackage{enumitem}
\usepackage{amsmath}

\usepackage{amssymb}

\usepackage{amsthm}
\usepackage[colorlinks=true,linkcolor=blue,citecolor=blue,urlcolor=blue]{hyperref}

\newtheorem{definition}{Definition}
\newtheorem{theorem}{Theorem}
\newtheorem{lemma}{Lemma}
\newtheorem{corollary}{Corollary}
\theoremstyle{remark}
\newtheorem{remark}{Remark}
\newtheorem{example}{Example}
\newtheorem{observation}{Observation}
\newtheorem{proposition}{Proposition}

\newcommand{\Hist}{\mathcal{H}}
\newcommand{\Spec}{\mathsf{Spec}}
\newcommand{\cost}{\mathsf{cost}}
\newcommand{\Cost}{\mathsf{Cost}}
\newcommand{\depth}{\mathsf{depth}}

\newcommand{\Depth}{\mathsf{Depth}}

\title{On the Complexity of Determinations}
\author{Joseph M. Hellerstein \\ \small{UC Berkeley and Amazon Web Services}}
\date{}

\begin{document}

\maketitle

\begin{abstract}
  Classical complexity theory measures the cost of computing a
  function, but many computational tasks require committing to
  one valid output among several.
  We introduce \emph{determination depth}---the minimum number
  of sequential layers of irrevocable commitments needed to
  select a single valid output---and show that no amount of
  computation can eliminate this cost.
  We exhibit relational tasks whose commitments are
  constant-time table lookups yet require exponential parallel
  width to compensate for any reduction in depth.
  A conservation law shows that enriching commitments merely
  relabels determination layers as circuit depth, preserving
  the total sequential cost.
  For circuit-encoded specifications, the resulting depth
  hierarchy captures the polynomial hierarchy
  ($\Sigma_{2k}^P$-complete for each fixed~$k$,
  PSPACE-complete for unbounded~$k$).
  In the online setting, determination depth is fully
  irreducible: unlimited computation between commitment
  layers cannot reduce their number.
\end{abstract}


\section{Introduction}

Classical complexity theory measures the cost of computing a
\emph{function}---producing the unique correct output.
But many tasks require resolving a \emph{relation}: choosing
one admissible outcome from a set, as when a text generator
selects one coherent continuation among many, or a
distributed system must agree on one valid output
among several.
We show that resolving a relation carries an intrinsic
sequential cost that no amount of computation can eliminate.
Classical hard functions are computationally expensive but 
have no cost on this axis.
However, there exist relational tasks in which every step is a
constant-time table lookup, yet any strategy using fewer
than $k$ sequential layers requires exponential width,
even with unbounded parallelism within each layer.
The bottleneck is not computation, nor communication, but
\emph{irrevocable commitment}.
This cost is invisible to classical complexity
measures---circuit depth, communication complexity,
adaptive queries---which track the cost of computing
or distributing an answer, not the cost of committing to one.

We call this choice process a \emph{determination}: a collection
of irrevocable commitments that narrow the admissible set
until a single outcome remains.
Commitments whose order affects which outcomes survive must
be sequenced into layers; those whose order does not matter
can be applied in parallel.
The \emph{determination depth} of a specification is the
minimum number of layers needed to resolve it---the
intrinsic cost of committing.

We establish three main results.
\emph{(1)~Exponential separation}
(Section~\ref{sec:exp-separation}): the depth--width
tradeoff is exponential, generalizing pointer
chasing~\cite{nisan1993rounds,mao2025gadgetless} from
functions to relations.
\emph{(2)~Computation cannot help}
(Section~\ref{sec:structural}): the sequential cost of
determination is governed by the dependency structure among
commitments.
Enriching the computational power of \emph{commitments}
merely relabels determination layers as circuit depth
within layers---the total sequential cost is conserved
(conservation law); the bound is tight.
Enriching the computational power of the \emph{strategy}
does not help either: under a commutative basis in the
online setting, an oracle characterization shows that
determination depth equals the minimum number of
observations of the growing history, regardless of
computation between observations.
\emph{(3)~Polynomial-hierarchy characterization}
(Section~\ref{sec:metacomplexity-body}): for circuit-encoded
specifications, ``is depth $\le k$?'' is
$\Sigma_{2k}^P$-complete for each fixed~$k$ and
PSPACE-complete for unbounded~$k$, certifying
determination depth as a semantic measure well-calibrated
against the classical hierarchy.

The exponential separation construction is autoregressive:
each commitment is a single irrevocable choice, and the
depth-$k$ lower bound says that $k$ sequential choices are
unavoidable regardless of parallelism.
This connects determination depth to chain-of-thought
reasoning, where each token of sequential inference is a
commitment layer (Appendix~\ref{sec:cot-appendix}).
Stable matching provides a complementary structural
connection: every finite determination depth arises as
the rotation-poset height of some instance
(Appendix~\ref{sec:stable-matching}).

The appendices apply the framework to
several classical settings, decomposing known costs into
intrinsic determination depth and circumstantial model
assumptions, and in several cases revealing new structure
that existing models do not capture.
Each determination layer carries a concrete physical cost: a
coordination round in distributed systems, a token of inference
in language models, a round of play in a
game. Existing models can hide this cost
by assuming away commitments, or inflate it by charging for
operations that resolve no ambiguity.

The appendices exhibit both phenomena.
BSP round complexity overcharges by counting communication
rounds that involve no commitment, and undercharges
by assuming fixed membership.
In extensive-form games, determination depth bounds the
game's \emph{strategic depth}---the number of moves where a
player must break a tie among equally-optimal options---which
can be strictly less than the game-tree depth; one player can
influence the other's strategic depth, creating a new
dimension of strategic interaction
(Appendix~\ref{sec:mechanism}).
In the distributed setting, the framework recovers the
Halpern--Moses impossibility of asynchronous common
knowledge~\cite{halpern1990knowledge}, with the
conservation law and exponential separation providing
quantitative bounds beyond the binary threshold.

\section{Framework}
\label{sec:specs}

We model both online and offline settings via
\emph{histories}---partially ordered sets of events that
represent everything that has occurred up to a given point.

\begin{definition}[History]
  A \emph{history} is a partially ordered set $H=(E,\to)$
  of events, where $\to$ is interpreted as causal
  precedence~\cite{lamport1978time}.
  Events are of two kinds:
  \emph{environment events} (inputs, messages, computation
  steps), which are given, and
  \emph{commitment events} (defined below), which are
  chosen.
\end{definition}

The results below are parametric in a class $\Hist$ of
histories under consideration.

\begin{definition}[History Extension]
  For histories $H_1=(E_1,\to_1)$ and $H_2=(E_2,\to_2)$ in
  $\Hist$, we write $H_1 \sqsubseteq H_2$ if
  $E_1 \subseteq E_2$ and $H_1$ is a downward-closed sub-poset
  of $H_2$ (i.e., whenever $e\in E_1$ and $e'\to e$ in $H_2$,
  then $e'\in E_1$).
\end{definition}

\noindent
$H_2$ extends $H_1$ by adding later events;
$\sqsubseteq$ is a prefix order representing information
growth.
A specification is \emph{online} if $\Hist$ contains
histories with proper extensions (the future is uncertain);
it is \emph{offline} if every history in $\Hist$ is maximal
under $\sqsubseteq$ (no further environment events will
occur).
The framework handles both uniformly.
An \emph{outcome} is an element of a set $O$ of externally
observable results.

\begin{definition}[Specification]
  A \emph{specification} is a mapping $\Spec : \Hist \to 2^O$
  assigning to each history a set of \emph{admissible outcomes}.
  We call $\Spec(H)$ the \emph{admissible set} at~$H$.
\end{definition}

\noindent
A specification may admit multiple outcomes for the same
history; resolving this ambiguity is the cost that
determination complexity measures.

\begin{definition}[Commitment]
  \label{def:commitment}%
  A \emph{commitment} $\varphi$ is a constraint on outcomes.
  We write $H \cdot \varphi$ for the history obtained by
  adding a commitment event $\varphi$ after all maximal
  events of~$H$
  (i.e., $e \to \varphi$ for every $e$ with no successor
  in~$H$), giving $H \sqsubseteq H \cdot \varphi$.
  Adding a commitment permanently narrows the admissible
  set (\emph{shrinkage}):
  $\Spec(H') \subseteq \Spec(H)$ for every
  $H' \sqsupseteq H \cdot \varphi$.
\end{definition}

\noindent
In the distributed setting, each agent's commitments
succeed only its own local maximal events, not the global
set; Appendix~\ref{sec:distributed-ck} develops this
extension fully.

Each commitment event $\varphi$ in a history carries its
context: the prefix of the history up to $\varphi$ determines
the admissible set at which $\varphi$ takes effect.
Environment events may change the admissible set in either
direction, but commitments can only shrink it, permanently
(shrinkage condition).
Not every relational specification requires commitments;
quantifying the requirement is the theme of this work.

\begin{definition}[Determination]
  \label{sec:determinations}%
  \label{def:refinement-determination}%
  The \emph{determination} $D(H)$ of a history $H$ over a
  commitment basis $\Phi$ is the sequence
  $\varphi_1 \cdot \cdots \cdot \varphi_m$
  ($\varphi_i \in \Phi$) of all commitment events in $H$,
  listed in causal order.
  Its \emph{cost} is $\cost(D(H)) \triangleq m$.
\end{definition}

\noindent
The commitment events in any single-agent history are totally
ordered (each follows all maximal events of its prefix), so
the determination is a well-defined sequence.
Writing $H_j$ for the prefix of $H$ through
$\varphi_j$, shrinkage accumulates:
$\Spec(H_j) \subseteq \Spec(H_{j-1})$ for each $j$.
A determination is \emph{valid} if each $\varphi_j$
preserves nonemptiness:
$\Spec(H') \neq \emptyset$ for every
$H' \sqsupseteq H_j$.
The goal is to narrow the admissible set to a single
outcome ($|\Spec(H)| = 1$); the formal notion of a
resolving strategy appears in
Definition~\ref{def:res-cost-depth}.

In the offline setting, environment events precede all
commitment events, so the determination is a single
consecutive block.
In the online setting, environment events may interleave
with commitments, breaking the determination into separate
runs; the effect of a commitment depends on which
environment events have occurred, so the online setting
requires an adaptive strategy
(formalized in Section~\ref{sec:algebra-full}).

\paragraph{Three forces on a history.}
\label{sec:three-forces}%
A history evolves under three distinct forces:
the \emph{environment} extends the history with new events
(input arrivals, message deliveries);
\emph{commitments} extend the history with events that
irrevocably narrow the admissible set;
and a \emph{determination strategy} selects which commitment
events to append, based on the history observed so far.
A determination strategy is a rule
$\sigma : \Hist \to \Phi^*$ mapping each observable history
to a sequence of commitments from the basis.
It is a semantic object: it specifies \emph{which} commitments
are made and \emph{when} in history, but says nothing about how they are
computed.
A protocol, algorithm, or agent is a realization of a
determination strategy; the complexity results in this paper
are stated at the strategy level.

\begin{example}[Consensus server]
  \label{ex:consensus}
  Consensus---agreeing on a single value---is a canonical
  coordination problem in distributed systems.
  A single sequential server that dictates a choice for the system
  is a simplified implementation, but sufficient to illustrate the
  three forces. The consensus specification maps a history to a set of
  candidate values.
  Proposals arrive over time (environment events,
  extending history $H \sqsubseteq H'$), each adding a value
  to the admissible set at the extended history.
  The strategy's determination consists of two commitments.
  First, $\varphi_c$ \emph{closes the proposal phase} at some history
  $H_c$: after $H_c \cdot \varphi_c$, the specification ignores
  subsequent proposals, so no environment extension of
  $H_c \cdot \varphi_c$ can enlarge the admissible set.
  Then $\varphi_w$ \emph{draws a winner}, resolving the 
  specification to one value
  from the closed proposal set.
  The draw depends on the closure---the strategy cannot
  choose a winner without knowing the candidates---so the
  two commitments require two sequential layers.
\end{example}

\paragraph{Commutation and depth.}
\label{sec:algebra}%
\label{def:commitment-commutation}%
Commitments $\varphi$ and $\psi$ \emph{commute with respect to}
$\Spec$ at $H$ if the order of extension does not matter:
$\Spec(H \cdot \varphi \cdot \psi) =
\Spec(H \cdot \psi \cdot \varphi)$.
Commutation is relative to a specification and may change
as the history grows: $\varphi$ and $\psi$ may commute at
$H$ but not at an extension $H' \sqsupseteq H$, or vice
versa.
This history-dependence does not arise in classical trace
theory~\cite{mazurkiewicz1977concurrent}, where independence
is static; in our setting, non-commuting commitments impose
irreducible sequencing precisely because their commutation
status can shift as the history evolves.
(A history records one linearization of the commitments;
commutation ensures the choice of linearization is
immaterial.)
The \emph{depth} of a determination is the minimum number
of \emph{commuting layers} needed to organize its
commitments: within each layer, all commitments commute
pairwise; across layers, they need not.
In the online setting, environment events may separate
consecutive commitments into distinct runs, and the depth
is the total number of layers across all runs
(formalized in Section~\ref{sec:algebra-full}).

\begin{definition}[Determination cost and determination depth]
  \label{def:res-cost-depth}
  A determination strategy $\sigma$ \emph{resolves} $\Spec$
  if, for every history $H$ arising under $\sigma$,
  $|\Spec(H)| = 1$ once all of $\sigma$'s commitments have
  been applied.
  Define
  $\Cost_\Phi(\Spec)
    \triangleq
    \min_{\sigma}\;
    \max_{H}\;
    \cost(D(H))$
  and
  $\Depth_\Phi(\Spec)
    \triangleq
    \min_{\sigma}\;
    \max_{H}\;
    \depth(D(H))$,
  where $\sigma$ ranges over strategies that resolve $\Spec$
  and $H$ ranges over histories arising under $\sigma$;
  or $\infty$ if no resolving strategy exists.
  In the offline setting (a single complete history), a
  strategy is simply a determination of that history, and
  the definition reduces to a min over determinations.
\end{definition}

\noindent
$\Cost_\Phi(\Spec)$ is the \emph{determination cost} and
$\Depth_\Phi(\Spec)$ the \emph{determination depth}.
Both are properties of the specification, not of any
particular strategy---the strategy is the witness, just as a
circuit witnesses the depth of a function.
Cost counts commitments; depth counts sequential layers.
In the online setting, the max over $H$ gives the
environment adversarial control bounded by the history
class~$\Hist$; this is a worst-case measure, and the
results in this paper focus on this setting.
In the consensus server example
(Example~\ref{ex:consensus}), the determination cost is~$2$
(close proposals, then draw) and the determination depth is
also~$2$: the draw depends on the closure, so the two
commitments cannot share a layer.

\subsection{Commitment bases}
\label{sec:bases}
\label{sec:intrinsic-depth}%

Determination depth depends on the commitment basis~$\Phi$.
We begin with the canonical minimal basis and then identify
three properties it satisfies that generalize to richer
settings.

\begin{definition}[Atomic basis]
  \label{def:atomic-basis}
  For a history $H$ and outcome $o$, the \emph{atomic
  commitment excluding $o$ at $H$} is the event $\varphi$
  that, when appended to $H$, excludes exactly $o$:
  $\Spec(H \cdot \varphi) = \Spec(H) \setminus \{o\}$,
  and $o \notin \Spec(H')$ for every
  $H' \sqsupseteq H \cdot \varphi$.
  It has no effect when appended to a history incomparable
  to $H$.
  The \emph{atomic basis} contains one such commitment per
  pair $(H, o)$.
\end{definition}

\noindent
Atomic commitments are the finest-grained irrevocable
commitments: each permanently excludes a single outcome at a
single history.

In online settings, depth under the atomic basis arises from
\emph{forward validity constraints}: two atomic
commitments made at history~$H$ cannot be applied in the
same layer if their joint effect empties the admissible
set at an extension of~$H$, as in the following
consensus example.
In this variant, environment events may also
\emph{shrink} the admissible set---e.g., by retracting
proposals:

\begin{example}[Three-valued consensus]
  \label{ex:three-valued}
  Consider a consensus instance
  (cf.\ Example~\ref{ex:consensus}) with three candidate
  values $O = \{a, b, c\}$ and
  $\Spec(H_0) = \{a,b,c\}$.
  Suppose the environment may extend $H_0$
  to a history
  $E \sqsupseteq H_0$ with $\Spec(E) = \{a,b\}$.
  The refined admissible sets
  $\Spec(H_0 \cdot \varphi_{\neg a})$ and
  $\Spec(E \cdot \varphi_{\neg a})$ are both non-empty
  ($\{b,c\}$ and $\{b\}$ respectively), so
  $\varphi_{\neg a}$ is valid at~$H_0$.
  But the joint refinement
  $\Spec(E \cdot \varphi_{\neg a} \cdot \varphi_{\neg b})$
  is empty: $\{a,b\} \setminus \{a,b\} = \emptyset$.
  Symmetrically, if the environment may also extend
  $H_0$ to $E'$ with $\Spec(E') = \{b,c\}$ and to $E''$
  with $\Spec(E'') = \{a,c\}$, then no pair of atomic
  commitments at $H_0$ is jointly valid, so resolution
  requires two layers: commit, observe the environment,
  then commit again.
\end{example}

\begin{definition}[Intrinsic determination depth]
  \label{def:intrinsic-depth}
  The \emph{intrinsic determination depth} of a specification
  $\Spec$ is
  $\Depth_\Phi(\Spec)$ where $\Phi$ is the atomic basis.
\end{definition}

\noindent
Intrinsic depth measures the irreducible sequential cost of
resolving a specification under the finest-grained
commitments, before any computational or architectural
assumptions are introduced.

\paragraph{Properties of the atomic basis.}
The atomic basis satisfies three properties that we now name
for use in the general theory.

\begin{definition}[Pointwise basis]
  \label{def:pointwise-basis}
  A commitment $\varphi$ is \emph{pointwise} if it acts as
  an independent filter on each outcome: whether
  $o \in \Spec(H \cdot \varphi)$ depends only on $o$
  and $H$, not on which other outcomes belong to $\Spec(H)$.
  A \emph{pointwise basis} consists entirely of pointwise
  commitments.
\end{definition}

\noindent
Pointwise commitments always commute (each outcome is filtered
independently, so application order is irrelevant).
Atomic commitments are pointwise, but the pointwise class is
strictly larger: a commitment that excludes multiple outcomes
simultaneously (e.g., ``keep only tuples with $v_1 = 3$'') is
pointwise but not atomic.
Commutativity can also hold for non-pointwise commitments:

\begin{definition}[Commutative basis]
  \label{def:commutative-basis}
  A commitment basis $\Phi$ is \emph{commutative} if every
  pair of commitments in $\Phi$ commutes at every history
  for every specification:
  $\Spec(H \cdot \varphi \cdot \psi)
    = \Spec(H \cdot \psi \cdot \varphi)$
  for all $\varphi, \psi \in \Phi$, all $\Spec$, and
  all $H \in \Hist$.
\end{definition}

\begin{definition}[Constant-depth basis]
  \label{def:natural-basis}
  A commitment basis $\Phi$ is \emph{constant-depth} for
  $\Spec$ if each commitment in $\Phi$ is computable in
  $O(1)$ circuit depth from a standard polynomial-size
  encoding of the current history.
\end{definition}

\noindent
The atomic basis is both commutative and constant-depth,
making it the canonical minimal basis for the online
setting where commitments are made before all information
has arrived.
The pointwise, commutative, and constant-depth properties
each appear in the results that follow.

\section{Exponential Depth--Width Separation}
\label{sec:exp-separation}

This section proves that determination depth is a meaningful
barrier: reducing it requires an exponential blowup in
parallel resources, even when each individual commitment is
trivial to compute.
The consensus server (Example~\ref{ex:consensus}) required
two sequential layers; we now exhibit tasks requiring $k$
layers for any~$k$.
We construct a family of constrained generation
tasks---modeling autoregressive text generation---in which
every commitment is a constant-time table lookup, yet any
strategy using $d' < k$ layers requires width at least
$(m/s)^{k-d'}$.

\subsection{The constrained generation task}
\label{sec:k-level-task}

Fix positive integers $k$ (number of positions), $m$ (domain size), and
$s$ (constraint sparsity), with $1 \le s \le m$.

\begin{definition}[$k$-position constraint chain]
  \label{def:k-level-chain}
  A \emph{$k$-position constraint chain} over domain~$[m]$ with
  sparsity~$s$ is a sequence of constraint functions
  $P_1 \subseteq [m]$ and
  $P_\ell \subseteq [m] \times [m]$ for $\ell = 2,\ldots,k$,
  such that $|P_1| = s$ and, for each $\ell \ge 2$ and each
  $a \in [m]$, the \emph{successor set}
  $P_\ell(a,\cdot) \triangleq \{b \in [m] : (a,b) \in P_\ell\}$
  satisfies $|P_\ell(a,\cdot)| = s$.
\end{definition}

\begin{definition}[$k$-position generation task]
  \label{def:k-level-task}
  Given a $k$-position constraint chain $(P_1,\ldots,P_k)$
  (encoded as environment events in an offline history), the
  \emph{constrained generation task} is to output any
  admissible tuple in
  \[
    \Spec_{\mathsf{cg}}(P_1,\ldots,P_k)
    \;\triangleq\;
    \bigl\{\,
    (v_1,\ldots,v_k) \in [m]^k
    \;\bigm|\;
    v_1 \in P_1
    \text{ and }
    (v_{\ell-1},v_\ell) \in P_\ell
    \text{ for } \ell=2,\ldots,k
    \,\bigr\}.
  \]
\end{definition}

\noindent
The specification is relational ($s^k$ admissible tuples when
$s \ge 2$) and models autoregressive generation: each position
is a token, $[m]$ is the vocabulary, and the $k - 1$
\emph{links} $1 \to 2 \to \cdots \to k$ encode local
coherence constraints.

\paragraph{Commitment basis.}
For each position $\ell \in \{1,\ldots,k\}$ and value
$v \in [m]$, the commitment $\varphi_{\ell,v}$ filters to
tuples with $v_\ell = v$.
These are pointwise, but the links introduce dependency:
committing to $v_\ell$ before $v_{\ell-1}$ may violate a
constraint.
Each commitment is a constant-time table lookup.

\begin{observation}[Sequential strategy]
  \label{lem:sequential-ub}
  Committing to one position per layer in order
  $v_1, v_2, \ldots, v_k$ always produces a valid tuple---each
  $v_\ell \in P_\ell(v_{\ell-1},\cdot)$ exists by
  sparsity---using $k$~layers and width~$1$.
\end{observation}

\subsection{Strategy Model and Lower Bound}

We formalize what a strategy with fewer than $k$ layers can do.

\begin{definition}[Uninformed link]
  \label{def:uninformed-link}
  Given a layer assignment $S_1,\ldots,S_{d'}$, a \emph{link}
  $\ell-1 \to \ell$ (for $\ell \in \{2,\ldots,k\}$) is \emph{informed} if
  position $\ell-1$ is assigned to a strictly earlier layer than position $\ell$,
  and \emph{uninformed} otherwise (i.e., $v_{\ell-1}$ is not
  yet determined when $v_\ell$ is committed).
\end{definition}

\begin{observation}[Uninformed-link count]
  \label{lem:uninformed-link-count}
  Any assignment of $k$ positions into $d' < k$ layers
  leaves at least $k - d'$ uninformed links.
\end{observation}

\begin{definition}[$d'$-layer strategy with width $w$]
  \label{def:strategy-model}
  A \emph{$d'$-layer strategy with width~$w$} partitions the
  $k$ positions into $d'$ groups $S_1,\ldots,S_{d'}$ and
  builds $w$ candidate tuples in parallel.
  In layer~$r$, the strategy assigns values to the positions
  in group $S_r$ for each candidate, depending only on the
  full input and earlier-layer assignments.
  The strategy resolves the specification if any candidate
  lies in $\Spec_{\mathsf{cg}}$; the final selection among
  valid candidates adds at most one layer, which does not
  affect the asymptotic lower bound.
\end{definition}

At each uninformed link, the strategy must guess a value
without knowing its predecessor; the theorem below shows
that each guess fails with probability controlled by a
distributional parameter~$\gamma$.

\begin{definition}[Conditionally $\gamma$-spread distribution]
  \label{def:gamma-spread}
  A distribution $\mathcal{D}$ over $k$-position constraint chains
  is \emph{conditionally $\gamma$-spread} if for every
  $\ell \in \{2, \ldots, k\}$, every $a, b \in [m]$, and every
  fixing of all constraint functions except the row
  $P_\ell(a, \cdot)$:
  $\Pr_{\mathcal{D}}[\,b \in P_\ell(a, \cdot)
    \mid \text{all other constraint data}\,]
    \le \gamma$.
\end{definition}

\noindent
In short, membership of any one element in a row has
probability $\le \gamma$, conditioned on the
rest of the chain.

\begin{definition}[Random constraint distribution]
  \label{def:random-ensemble}
  In the \emph{random $(k,m,s)$-distribution}, every row
  ($P_1$ and each $P_\ell(a,\cdot)$) is an independent
  uniformly random $s$-subset of~$[m]$.
  This distribution is conditionally $(s/m)$-spread.
\end{definition}

\noindent
Over a distribution on inputs, a strategy may resolve some
constraint chains but not others; we measure the probability
of resolution.

\begin{theorem}[Exponential depth--width separation]
  \label{thm:exp-separation}
  Let $\mathcal{D}$ be a conditionally $\gamma$-spread
  distribution over $k$-position constraint chains with
  $\gamma < 1$.
  \begin{enumerate}[label=(\alph*)]
    \item A sequential strategy with $k$ layers and width~$1$
          resolves $\Spec_{\mathsf{cg}}$ with probability~$1$.
    \item For any $d' < k$, any deterministic $d'$-layer
          strategy with width~$w$ resolves
          $\Spec_{\mathsf{cg}}$ with probability at
          most $w \cdot \gamma^{\,k-d'}$.
    \item Consequently, achieving constant resolution
          probability with $d' < k$ layers requires width
          $w \ge (1/\gamma)^{k-d'}$.
  \end{enumerate}
  For the random $(k,m,s)$-distribution, $\gamma = s/m$, giving
  width $w \ge (m/s)^{k-d'}$.
\end{theorem}

The model is strictly more permissive than autoregressive
generation (which fixes left-to-right order and
width~$1$): it allows arbitrary layer assignments,
arbitrary within-layer parallelism, and $w$ independent
candidates.
The only constraint is that commitments within a layer
cannot depend on same-layer outcomes---the defining
property of parallel execution---so the lower bound
applies to every parallel strategy that respects this
constraint, including beam search, speculative decoding,
and diffusion-style refinement.

Part~(a) is Observation~\ref{lem:sequential-ub}.
Part~(c) follows immediately from part~(b), which we prove now:

\begin{proof}[Proof of Theorem~\ref{thm:exp-separation}(b)]
  Fix a deterministic $d'$-layer strategy with width~$w$ and
  layer assignment $S_1,\ldots,S_{d'}$.
  By Observation~\ref{lem:uninformed-link-count}, at least $t \ge k - d'$
  links $\ell_1, \ldots, \ell_t$ are uninformed.

  Fix a single candidate tuple $(v_1,\ldots,v_k)$ and
  consider an uninformed link $\ell_j$.
  The strategy picks $v_{\ell_j}$ without seeing
  $P_{\ell_j}(v_{\ell_j - 1}^*, \cdot)$, which is still
  $\gamma$-spread by assumption, so
  $\Pr[(v_{\ell_j - 1}^*, v_{\ell_j}) \in P_{\ell_j}]
    \le \gamma$.
  Multiplying across all $t \ge k - d'$ links gives
  success probability at most $\gamma^{k-d'}$ per candidate;
  a union bound over $w$ candidates gives the result.
\end{proof}

\begin{corollary}[Hard instances for any fixed strategy]
  \label{cor:hard-instances}
  For every $k, m, s$ with $m \ge 2s$, every $d' < k$, and
  every (possibly randomized) $d'$-layer strategy with
  width~$w$, there exists a constraint chain on which the
  strategy resolves $\Spec_{\mathsf{cg}}$ with probability
  at most
  $w \cdot (s/m)^{k-d'}$.
\end{corollary}

\noindent
Theorem~\ref{thm:exp-separation}(b) bounds the expected
resolution probability of any deterministic strategy over the
random $(k,m,s)$-distribution by $w \cdot (s/m)^{k-d'}$,
so for each deterministic strategy some chain exhibits this
bound.
Yao's minimax principle~\cite{yao1977probabilistic} gives
the same guarantee for randomized strategies.

\begin{remark}[Approximate determination and the exchange rate]
  \label{sec:approx-determination}
  When a strategy tolerates constraint violations, the
  expected violation count is at least
  $(k - d')(1 - \gamma)$ by Theorem~\ref{thm:exp-separation};
  for the random $(k,m,s)$-distribution this is tight.
  The spread parameter $\gamma$ controls the exchange rate
  between distributional structure and sequential depth:
  a distribution with more exploitable structure has a
  smaller effective $\gamma$, reducing the layers needed.
  At one extreme ($\gamma \to 0$, perfect prediction), no
  layers are needed; at the other ($\gamma = s/m$, no
  exploitable structure), the full $k$ layers are required.
\end{remark}

\begin{remark}[Generality beyond autoregressive generation]
  \label{rem:diffusion}
  The strategy model applies to any procedure operating in
  sequential rounds of parallel refinement, including
  diffusion models.
  A diffusion model with $T$ denoising steps is a $T$-round
  strategy; if $T < k$, the uninformed-link argument applies
  and the exponential lower bound holds with $d' = T$.
\end{remark}

\subsection{Distributed Extension and Pointer Chasing}

The constrained generation task shares the same
constraint-chain structure as $k$-step pointer
chasing~\cite{nisan1993rounds,papadimitriou1984communication},
differing only in sparsity~$s$, and plays the same role for
determination complexity that pointer chasing plays for
communication complexity.
The following tradeoff extends the width bound to a
distributed setting where communication can substitute for
width.
Consider a $k$-party communication model in which the
constraint chain is distributed: player~$\ell$ holds
$P_\ell$ privately and sends $b_\ell$ bits to a central
referee, who must output an admissible tuple.

\begin{theorem}[Depth--width--communication tradeoff]
  \label{thm:three-way}
  For the random $(k,m,s)$-distribution in the $k$-party model,
  any $d$-round protocol with width~$w$ and per-player
  communication $b_\ell$ bits satisfies
  $\log w + \sum_{\ell \in U} b_\ell
    \ge |U| \cdot \log(m/s)$,
  where $U$ is the set of uninformed links ($|U| \ge k - d$).
\end{theorem}

\noindent
Setting $b_\ell = 0$ recovers the width bound; setting $w = 1$
gives a communication lower bound.
Both are tight.
Proof in Appendix~\ref{sec:conservation-tight-proof}.

\begin{observation}[Pointer chasing as degenerate determination]
  \label{prop:pointer-chasing}
  At sparsity $s = 1$ the specification is functional
  (determination depth~$0$); the classical round complexity
  of $k$-step pointer
  chasing~\cite{nisan1993rounds} arises entirely from
  distributing the chain across agents, not from relational
  choice.
  At $s \ge 2$ the specification is relational (determination
  depth~$k$); the exponential width bound applies even to a
  centralized machine.
  The transition occurs exactly at $s = 1$ vs.\ $s \ge
  2$---the boundary between functions and relations.
\end{observation}

\section{Structural Results}
\label{sec:structural}
\label{sec:oracle-body}

The exponential separation shows that determination depth is a
meaningful complexity measure with concrete consequences.
We now establish structural results that explain \emph{why}
the separation works and what it implies for the general
theory.
The sequential cost of determination is governed by the
dependency structure among commitments---a structure that
computation cannot alter.
The \emph{conservation law}
(Theorem~\ref{obs:conservation}) shows that enriching the
commitment basis merely relabels determination layers as
circuit depth within layers---the total sequential cost is
conserved.
The \emph{oracle characterization}
(Proposition~\ref{thm:oracle-depth-body}) shows that under a
commutative basis in the online setting, even unbounded
computation between layers cannot reduce their number.
Communication re-enters in the distributed setting
(Appendix~\ref{sec:distributed-ck}).

\subsection{Conservation of sequential depth}
\label{sec:conservation-body}

A specification may have internal dependency structure---a
chain of commitments where each depends on the previous
one's outcome---even in the offline setting, where forward
validity constraints are absent.
A richer basis can reduce the number of determination
layers $d$ by bundling multiple constant-depth commitments
into a single commitment, but evaluating that commitment
requires circuit depth $c_i$ that accounts for the bundled
dependencies.
The total sequential depth of commitment
evaluation---summing $(1 + c_i)$ over layers---cannot fall
below $\Depth_{\Phi_0}(\Spec)$.
This is not a counting identity: a richer basis could
in principle exploit algebraic cancellations to resolve the
specification in fewer total sequential steps than the
longest dependency chain, and the following theorem's
content is that no such shortcut exists.

\begin{theorem}[Conservation law for sequential depth]
  \label{obs:conservation}
  Let $\Phi_0$ be a constant-depth basis for $\Spec$
  (Definition~\ref{def:natural-basis}).
  For any basis $\Phi$ and any valid determination over $\Phi$
  that resolves $\Spec$ in $d$ layers, let $c_i$ be the
  circuit depth of layer~$i$ (where the circuit's input is
  the current history, which includes all prior commitment
  events and their effects on the admissible set).
  Then
  $
    \sum_{i=1}^{d}(1 + c_i)
    \;\ge\; \Depth_{\Phi_0}(\Spec).
  $
\end{theorem}

Each edge $\varphi_j \to \varphi_{j+1}$ in the
forced-dependency DAG of $\Phi_0$ has a \emph{witness
exclusion}: a triple $(H_j, H_j \cdot \varphi_j, o_j)$
such that excluding $o_j$ at $H_j$ is invalid, but
excluding $o_j$ at $H_j \cdot \varphi_j$ is valid.
The proof traces the longest path in this DAG and shows
that each edge must be ``paid for'' by either a layer
boundary or a circuit path within a layer.
The key step is:

\begin{observation}[Semantic dependency implies circuit dependency]
  \label{lem:semantic-circuit}
  If $\varphi_{j+1}$ is forced to depend on $\varphi_j$
  (under a constant-depth basis $\Phi_0$), and a circuit $C$
  computes a layer that achieves the witness exclusions of
  both, then $C$ contains a directed path from the gates
  for $\varphi_j$'s exclusion to those for
  $\varphi_{j+1}$'s.
\end{observation}

\begin{proof}
  By forced dependency, $\varphi_{j+1}$'s exclusion is
  invalid unless $\varphi_j$'s has occurred.
  If $C$ had no such path, $\varphi_{j+1}$'s output would
  be independent of $\varphi_j$'s.
  But the witness exclusion guarantees instances where
  excluding $o_{j+1}$ without first excluding $o_j$ is
  invalid, so $C$ would produce an invalid refinement on
  such an instance---a contradiction.
\end{proof}

\begin{proof}[Proof of Theorem~\ref{obs:conservation}]
  Let $d^* = \Depth_{\Phi_0}(\Spec)$ and let
  $\varphi_1 \to \cdots \to \varphi_{d^*}$ be a longest
  forced-dependency path in the $\Phi_0$-dependency DAG
  (Proposition~\ref{thm:dependency-chain-lb}).
  Each edge $\varphi_j \to \varphi_{j+1}$ has a
  witness exclusion.
  Fix any determination under basis $\Phi$ with $d$ layers
  and per-layer circuit depths $c_1, \ldots, c_d$.
  Since the determination resolves $\Spec$, each $o_j$ is
  eventually excluded; let $\lambda(j)$ be the $\Phi$-layer
  in which $o_j$ is first excluded from the admissible set
  at all extensions.

  For each consecutive pair $\varphi_j \to \varphi_{j+1}$,
  exactly one of the following holds.
  \emph{Case~1:} $\lambda(j) < \lambda(j+1)$---the pair
  crosses a layer boundary, contributing at least~$1$ to $d$.
  \emph{Case~2:} $\lambda(j) = \lambda(j+1) = i$---both
  exclusions occur in the same layer; by forced dependency,
  the circuit computing layer~$i$ must contain a directed
  path from the gates for $o_j$'s exclusion to those for
  $o_{j+1}$'s, contributing at least~$1$ to $c_i$.
  \emph{Case~3:} $\lambda(j) > \lambda(j+1)$---the
  successor is excluded before the predecessor, contradicting
  forced dependency.

  Each of the $d^* - 1$ pairs contributes at least~$1$ to
  either $d$ or some~$c_i$, and each layer visited
  contributes its $1$ term.
  Hence $\sum_{i=1}^{d}(1 + c_i) \ge d^*$.
\end{proof}

\noindent
The constrained generation task
(Section~\ref{sec:exp-separation}) witnesses tightness: for
every split of the $k$ dependency links, the bound is
achieved with equality
(Appendix~\ref{sec:conservation-tight-proof}).

\subsection{Oracle characterization}
\label{sec:oracle-characterization}
\label{sec:coordination}

The conservation law shows that enriching commitments
cannot reduce the total sequential cost.
What if instead the \emph{strategy} has unbounded
computational power between commitments?
In the online setting, this question is non-trivial:

\begin{proposition}[Oracle power does not reduce depth]
  \label{thm:oracle-depth-body}
  \label{thm:oracle-depth}
  \label{lem:oracle-layer}
  Under a commutative basis~$\Phi$ in the online setting,
  grant the strategy a free call to a disclosure oracle
  at the start of every layer
  (a function that may inspect the entire history, perform
  unbounded computation, and return any value).
  Any determination strategy that resolves $\Spec$ still
  requires at least $\Depth_\Phi(\Spec)$ layers.
\end{proposition}

\noindent
Under a commutative basis, the only constraint on
co-applying commitments within a layer is \emph{forward
validity}: their combined effect must preserve nonemptiness
of the admissible set at all extensions.
Forward validity is a property of the specification, not of
the strategy's knowledge, so no oracle can increase the
number of commitments that fit in one layer
(Example~\ref{ex:three-valued}).
Commutativity is essential: without it, non-commuting
commitments can be sequenced within a single layer,
collapsing multiple algebraic layers.

\subsection{Relationship to computational complexity}

Determination depth measures a cost that classical
complexity does not track.
A function has determination depth~$0$ regardless of its
circuit depth; a relation can require exponential width to
compensate for any reduction in depth, even when every
individual commitment is a constant-time operation
(Theorem~\ref{thm:exp-separation}).
This width cost is purely relational: at sparsity $s = 1$
(a function), it vanishes entirely
(Observation~\ref{prop:pointer-chasing}).
In the online setting, the oracle characterization shows
that determination depth is irreducible: no amount of
computation between layers can reduce their number.

Classical computational complexity plays a narrow role
\emph{within} the framework: each commitment is itself a
function, so its evaluation has classical circuit depth.
The conservation law shows that this is all computation can
do---it emulates atomic-basis determination layers
one-for-one as
circuit layers within a richer commitment, preserving the
total sequential cost exactly.
Determination depth under the atomic basis is thus a
universal lower bound on the total sequential cost of
resolution, regardless of the basis used.

Determination depth on the atomic basis therefore provides
a complete accounting of the sequential cost of resolving a
relation---both commitment dependency and forward
validity, in both offline and online settings.
Computational depth captures only the cost of evaluating
individual commitments, a sub-problem that arises when the
basis is enriched and only in settings where forward
validity is absent.

In the distributed setting, the oracle characterization
serves as a single-agent baseline:
Appendix~\ref{sec:distributed-ck} shows that asynchronous
multi-agent systems cannot always achieve this baseline,
recovering the Halpern--Moses impossibility of common
knowledge as a special case and providing quantitative
depth predictions beyond the binary threshold.

\section{Metacomplexity of Determination Depth}
\label{sec:metacomplexity-body}

The preceding sections establish determination depth as a
complexity measure with concrete separations and structural
laws.
We now ask: how hard is it to \emph{compute} the determination
depth of a given specification?
The answer ranges from NP-hard to PSPACE-complete depending
on the setting, with the polynomial hierarchy
arising as the exact hierarchy of determination depths.

\begin{proposition}[NP-hardness in the offline setting]
  \label{thm:metacomplexity-body}
  Computing determination depth is NP-hard in the offline
  setting: for decision tree synthesis, determination depth
  equals minimum tree depth, which is NP-hard to
  compute~\cite{hyafil1976constructing}.
\end{proposition}

\noindent
Proof in Appendix~\ref{sec:metacomplexity}.
Depth here reflects the tree structure: a depth-$d$
decision tree has $d$ levels of nodes, and each level is
one round of variable-test commitments.

In the online setting, an adversarial environment adds
universal quantification---the environment fixes some
variables, the determiner fixes others---and the
metacomplexity captures the full polynomial hierarchy.

\begin{theorem}[Determination depth captures the polynomial
  hierarchy]
  \label{thm:ph-characterization-body}
  Given a Boolean formula $\theta(x_1, \ldots, x_n)$
  encoded by a polynomial-size circuit, the outcome set is
  $O = \{0,1\}^n$ (all variable assignments) and the
  initial admissible set is $\Spec(H_0) = O$.
  The commitment basis consists of pointwise filters
  ``set $x_i = b$''; the environment fixes some variables
  adversarially, the determiner fixes the rest.
  After all $n$ variables are fixed, the strategy succeeds
  iff the resulting assignment satisfies~$\theta$.
  \begin{enumerate}[label=(\roman*),nosep]
    \item For each fixed $k$, ``is determination
          depth $\le k$?'' is $\Sigma_{2k}^P$-complete.
    \item For unbounded $k$ (given as input), the problem is
          PSPACE-complete.
  \end{enumerate}
\end{theorem}

\begin{proof}[Proof sketch]
  \emph{Upper bound.}
  The determiner guesses $k$ rounds to control; the
  environment controls the remaining $n-k$.
  Each round, the controlling player fixes one variable.
  The determiner's choices are existential quantifiers,
  the environment's are adversarial (universal).
  In the worst case these alternate
  ($\exists\forall\exists\forall\cdots$), and consecutive
  same-player rounds collapse, giving at most $2k$
  alternations ($\Sigma_{2k}^P$); for unbounded~$k$,
  PSPACE.

  \emph{Hardness.}
  Reduce from $\Sigma_{2k}$-QBF for fixed~$k$, or from
  TQBF for unbounded~$k$.
  Construct a specification with $2k$ variables: at odd
  rounds the determiner fixes an existential variable, at
  even rounds the environment fixes a universal variable.
  The admissible set starts as $O$; each round narrows it
  by fixing one variable.
  The determiner has a depth-$k$ strategy iff the
  corresponding QBF is true.
  Full proof in Appendix~\ref{sec:ic}, where the game is
  formalized as a QBF instance with variable-fixing
  commitments.
\end{proof}

\noindent
The polynomial hierarchy is thus precisely the hierarchy of
determination depths for circuit-encoded specifications.
(Offline, the
metacomplexity is NP-hard but does not capture the full
hierarchy; see Proposition~\ref{thm:metacomplexity-body}.)
Where circuit depth and communication complexity parameterize
the PH through computational resources, determination depth
parameterizes it through a semantic one---layers of irrevocable
choice.
The correspondence is not the contribution; it is a
calibration.
An alternating Turing machine has a fixed quantifier prefix
determined by the program; here the determiner optimizes
\emph{which} variables to control and in which order, so the
quantifier schedule is itself part of the problem.
The alternation arises from the interaction between
commitments and environment extensions, not from the
encoding of the specification.
More fundamentally, the paper's main contributions---the
exponential separation
(Theorem~\ref{thm:exp-separation}), the conservation law
(Theorem~\ref{obs:conservation}), and the oracle
characterization
(Proposition~\ref{thm:oracle-depth-body})---establish that
determination depth resists computation in both the offline
and online settings, a phenomenon not captured by
quantifier alternation.

\section{Conclusion}
\label{sec:conclusion}

Complexity theory has measured the cost of computing a uniquely
determined answer.
This paper introduces determination depth to measure the
cost of \emph{committing} to an answer when many are
admissible---a cost that no amount of computation can
eliminate.
The exponential separation
(Theorem~\ref{thm:exp-separation}) shows this cost is real,
the conservation law
(Theorem~\ref{obs:conservation}) and oracle characterization
(Proposition~\ref{thm:oracle-depth-body}) show that
computation cannot reduce it,
and the PH characterization
(Theorem~\ref{thm:ph-characterization-body}) shows it is
well-calibrated against the classical hierarchy.
In the distributed setting, the framework recovers the
Halpern--Moses impossibility of asynchronous common
knowledge as a special case
(Appendix~\ref{sec:distributed-ck}),
and the depth--width--communication tradeoff
(Theorem~\ref{thm:three-way}) shows that the depth, width,
and communication costs of resolving a specification are
fungible.
The appendices ground the framework in BSP round complexity,
chain-of-thought reasoning, stable matching, extensive-form
games, and distributed graph coloring;
Appendix~\ref{sec:related-work} surveys related work in
detail.
A number of directions are open:

\emph{Distributional vs.\ worst-case hardness.}
The exponential separation is distributional: for any single
fixed chain, a strategy with full knowledge resolves in one
layer.
Characterizing which restrictions on the strategy's knowledge
or power yield worst-case hardness is open.

\emph{Fault tolerance and randomization.}
In every deterministic, fault-free multi-party setting we
have examined, the specification's semantic structure fully
accounts for the known round complexity via determination
depth
(Appendices~\ref{sec:bsp-diagnosis}--\ref{sec:local-diagnosis}).
The scope of this correspondence is open:
fault-tolerant settings and randomized protocols lie outside
the current framework, though the Lov\'{a}sz Local Lemma
already shows that randomization can collapse determination
depth in some settings
(Appendix~\ref{sec:local-diagnosis}).

\emph{Reversible commitments.}
The framework assumes commitments are irrevocable.
In many settings---autoregressive generation with
backtracking, transactional rollback, speculative
execution---commitments can be reversed at a cost.
Extending the commitment algebra to include inverses
(a group rather than a monoid) and
characterizing the resulting depth--cost tradeoffs is open.

\emph{Provenance of determinations.}
Classical data provenance explains query answers via a
commutative semiring over monotone
derivations~\cite{green2007provenance}.
Determinations introduce a different algebra: a
conditionally commutative monoid of layered commitments.
The two structures do not align naturally.
A unifying algebraic theory of \emph{determination provenance}---
tracking how commitments at each layer shape downstream
explanations---remains open.

\appendix

\section{Formal Definitions for Determination Depth}
\label{sec:algebra-full}
\label{sec:normal-forms}

This appendix collects the formal definitions deferred from
Section~\ref{sec:algebra}; no new results are claimed.
Fix a commitment basis~$\Phi$ and a specification~$\Spec$
throughout.

The goal is to characterize the parallelism available in a
determination $D(H)$ (the commitment subsequence of a
history~$H$).
In the \emph{offline} setting, all commitment events are
consecutive (no environment events interleave), so the
entire determination can be analyzed as a single sequence:
adjacent commitments that commute can be parallelized into
layers, and the depth of $D(H)$ is the minimum number of
such layers.
In the \emph{online} setting, environment events may
separate consecutive commitments, and commutation is only
meaningful for commitments that share the same history
prefix.
The layering then applies to each maximal run of consecutive
commitments (between successive environment events), and the
depth of $D(H)$ is the total number of layers across all
such runs.
$\Depth_\Phi(\Spec)$ is defined
(Definition~\ref{def:res-cost-depth}) as the minimum
worst-case depth over all strategies that resolve $\Spec$.

\subsection{Commuting Layers and Depth}
\label{sec:depth}

We extend the $\cdot$ notation to sets: $H \cdot L$
is the history obtained by appending all commitments in $L$
after~$H$ (in any order).

\begin{definition}[Commuting layer]
  A finite set $L$ of commitments from a common basis
  $\Phi$ is a \emph{commuting layer for $\Spec$ at $H$} if
  applying its commitments consecutively after $H$ (with no
  intervening environment events) is order-independent:
  for any two listings
  $\psi_1,\ldots,\psi_t$ and $\psi'_1,\ldots,\psi'_t$ of
  the elements of $L$,
  \[
    \Spec(H \cdot \psi_1 \cdot \psi_2 \cdots \psi_t)
    \;=\;
    \Spec(H \cdot \psi'_1 \cdot \psi'_2 \cdots \psi'_t).
  \]
  Pairwise commutation
  (Definition~\ref{def:commitment-commutation}) at $H$ is
  sufficient; every basis used in this paper satisfies this
  condition.
\end{definition}

\begin{definition}[Layering and determination depth]
  \label{def:layering-depth}
  Let $D(H)=\varphi_1 \cdot \cdots \cdot \varphi_m$ be
  the determination of a history~$H$ over~$\Phi$.
  A \emph{run} is a maximal subsequence of consecutive
  commitment events in $H$ (with no intervening environment
  events).
  A \emph{layering of a run} $R = \psi_1 \cdot \cdots
  \cdot \psi_r$ is a partition into nonempty sets
  $L_1,\ldots,L_k$ such that each $L_i$ is a commuting layer
  for $\Spec$ at the history $H_0 \cdot L_1 \cdots L_{i-1}$,
  where $H_0$ is the history prefix immediately before the
  run.
  The \emph{depth} of a run is the minimum $k$ over all its
  layerings.
  The \emph{depth} of $D(H)$, written $\depth(D(H))$, is
  the sum of the depths of its runs.
  In the offline setting, the entire determination is a
  single run.
\end{definition}

\noindent
Every run admits a depth-minimal layering; in the offline
setting (a single run), we call this the \emph{layered
normal form}.
Depth is bounded by cost ($\depth(D) \le \cost(D)$, since
layers are nonempty), but can be much smaller when most
commitments commute.

\subsection{Dependency Chains and Depth Characterization}

To lower-bound depth, we exhibit commitments that cannot share
a layer.

\begin{definition}[Forced dependency]
  \label{def:forced-dependency}
  Let $\varphi$ and $\psi$ be commitment events in a
  determination $D(H)$.
  We say $\psi$ \emph{locally depends on} $\varphi$ (in $H$)
  if, in every layering of $D(H)$, $\varphi$ appears in a
  strictly earlier layer than $\psi$.
  We say $\psi$ \emph{universally depends on} $\varphi$ if
  this holds in every resolving history whose determination
  contains both (used only in the distributed setting,
  Appendix~\ref{sec:distributed-ck}).
\end{definition}

\noindent
In practice, forced dependency arises when $\psi$ acts as the
identity (or violates validity) until $\varphi$ has been
applied.
The stable matching application
(Section~\ref{sec:examples}) exhibits this pattern.

\begin{remark}[Online depth and forced dependency]
  \label{rem:online-depth}
  The offline results (Proposition~\ref{thm:dependency-chain-lb}
  below) use local forced dependency.
  The distributed results
  (Theorem~\ref{thm:async-impossibility}) use universal
  forced dependency.
  In the online setting, the oracle characterization
  (Proposition~\ref{thm:oracle-depth-body}) handles depth
  directly, without reducing to dependency chains.
\end{remark}

\begin{definition}[Dependency chain]
  \label{def:dependency-chain}
  A sequence $\varphi_1,\varphi_2,\ldots,\varphi_k$ is a
  \emph{dependency chain of length $k$} if each
  $\varphi_{i+1}$ locally depends on $\varphi_i$.
\end{definition}

\begin{proposition}[Depth equals longest dependency chain
  (offline)]
  \label{thm:dependency-chain-lb}%
  \label{cor:depth-invariant}%
  In the offline setting,
  $\Depth_\Phi(\Spec)$ equals the maximum length of a
  dependency chain over all resolving determinations.
  In the online setting, the longest chain is a lower bound
  on depth; run boundaries may force additional layers.
\end{proposition}

\begin{proof}
  A dependency chain of length~$k$ requires $k$ layers
  (each successive commitment must appear in a strictly
  later layer), so $\Depth_\Phi(\Spec) \ge$ the longest
  chain in both settings.
  In the offline setting (a single run), any determination
  induces a dependency DAG on its commitments (draw an edge
  $\varphi \to \psi$ whenever $\psi$ depends on $\varphi$).
  A depth-minimal layering corresponds to a minimum-height
  topological layering of this DAG, whose height equals the
  length of the longest chain.
\end{proof}

\section{Distributed Determinations and Common Knowledge}
\label{sec:distributed-ck}

This appendix extends the oracle characterization
(Section~\ref{sec:oracle-characterization}) to the
distributed setting, connecting determination depth to
synchronization and common knowledge.

\paragraph{Distributed setting.}
We extend the framework to multiple agents.
Fix a finite set of agents $\{1, \ldots, n\}$.
Each event in a history $H$ is associated with an agent;
for agent~$p$, the \emph{projection} $H|_p$ is the
restriction of~$H$ to $p$'s events---local computation
steps, message sends, and message receives at~$p$---ordered
by the transitive closure of~$\to$ restricted to $p$'s
events: $e_1$ precedes $e_2$ in $H|_p$ whenever
$e_1 \to^* e_2$ in $H$ and both are events at~$p$.
(Causal paths through other agents' events are not
directly visible, but their ordering effects are preserved.)
The \emph{frontier} of $H|_p$ is the set of sinks of
$H|_p$ (events with no successors under $\to$).

\paragraph{Indistinguishability.}
In the Halpern--Moses
framework~\cite{halpern1990knowledge}, two global states are
indistinguishable to agent~$p$ if $p$'s local state is the
same in both.
In our formalism, the analogue of local state is the
projection $H|_p$.
Two global histories $H, H'$ are
\emph{$p$-indistinguishable}, written $H \sim_p H'$, if
$H|_p$ and $H'|_p$ are isomorphic as partial orders of
typed events---that is, there exists an order-preserving
bijection between the events of $H|_p$ and $H'|_p$ that
preserves event types (message contents, commitment
values, and local computation steps).

\paragraph{Local knowledge and mutual knowledge.}
A global property $S$ (a set of histories) is \emph{known}
by agent~$p$ at $H$ if $S$ holds at every history
$p$-indistinguishable from $H$:
$K_p(S, H) \;\Leftrightarrow\; \forall H'\!: H \sim_p H'
\Rightarrow H' \in S$.
$S$ is \emph{known to everyone} at $H$ if $K_p(S, H)$ for
every~$p$; write $E(S, H)$ for this.
Define \emph{$k$-th order mutual knowledge} inductively:
$E^0(S, H)$ holds iff $H \in S$;
$E^{k+1}(S, H)$ holds iff for every agent~$p$ and every
$H'$ with $H \sim_p H'$, $E^k(S, H')$ holds.
\emph{Common knowledge} of $S$ at $H$ means $E^k(S, H)$ for
all $k \ge 0$.

\paragraph{Synchronization points.}
We assume each agent's projection $H|_p$ is a chain (a total
order on $p$'s events); this holds whenever each agent
executes sequentially, as in standard distributed computing
models.

A determination strategy may invoke multiple synchronization
points in sequence.
The \emph{$j$-th synchronization point}
is a set of distinguished events
$\{e^*_{j,1}, \ldots, e^*_{j,n}\}$, one local event per
agent (i.e., $e^*_{j,p}$ is in the projection of~$p$),
satisfying two conditions:
(i)~the events form a \emph{consistent cut}: no
$e^*_{j,p}$ causally follows any $e^*_{j,q}$ for
$p \neq q$ (so the global history truncated to the cut is
well-defined); and
(ii)~the admissible set at the cut is the same as seen by
each agent---$\Spec(H_{\le j})$ is independent of which
agent's perspective is used, where $H_{\le j}$ is the
global history through the cut.

At a synchronization point, every agent sees the same
admissible set (condition~(ii)), and each agent knows the
synchronization occurred ($e^*_{j,p}$ is in its
projection).

\begin{lemma}[Synchronization establishes common knowledge]
  \label{lem:sync-ck}
  Assume the synchronization-point conditions
  (i)--(ii) are common knowledge among all agents.
  Then at each synchronization point, the current admissible
  set is common knowledge.
\end{lemma}

\begin{proof}
  Condition~(ii) gives $E(S, H)$; common knowledge of the
  conditions iterates this to all levels.
\end{proof}

\paragraph{Connection to the oracle framework.}
The oracle characterization
(Section~\ref{sec:oracle-characterization}) assumes a single
global oracle that sees the entire history.
In the distributed setting, each agent has access only to a
\emph{local} oracle that sees its projection~$H|_p$.
A synchronization round (conditions (i)--(ii)) is precisely
the mechanism that elevates local oracles to global ones:
after synchronization, every agent's local state includes
the same admissible set (by condition~(ii) and
Lemma~\ref{lem:sync-ck}), so a local oracle can determine
the admissible set without global access.
Between synchronization points, local oracles are strictly
weaker---each sees only its own projection.

\begin{proposition}[Determination depth equals
  synchronization points]
  \label{prop:ck-rounds}
  In the online setting under a commutative basis~$\Phi$
  with $n \ge 2$ agents, if every layer boundary involves a
  cross-agent forced dependency
  (Definition~\ref{def:forced-dependency}), then
  $\Depth_\Phi(\Spec)$ equals the minimum number of
  synchronization points needed to resolve $\Spec$.
  When some layers involve only local dependencies, the
  minimum number of synchronization points may be smaller
  than $\Depth_\Phi(\Spec)$, but
  $\Depth_\Phi(\Spec)$ synchronization points always
  suffice.
\end{proposition}

\begin{proof}
  \emph{Lower bound.}
  Between two consecutive synchronization points, each agent
  acts on its projection $H|_p$ (its local events, received
  messages, and environment inputs).
  Within a single layer, an agent can safely apply its own
  commitments without synchronization: by definition, the
  commitments in a layer commute and are jointly valid, so
  each agent's share can be applied independently.
  The constraint arises at layer boundaries.
  By the cross-agent assumption, each layer boundary has a
  commitment in layer $i+1$ at some agent~$q$ that
  universally depends
  (Definition~\ref{def:forced-dependency}) on a commitment
  in layer~$i$ at a different agent~$p$.
  Without a synchronization point, $q$ cannot verify
  that $p$ has completed its layer-$i$
  commitment: $q$'s projection may be consistent
  with histories in which $p$ has not yet acted.
  Applying a layer-$(i+1)$ commitment in such a history
  risks invalidity.
  A synchronization point establishes common knowledge that
  layer~$i$ is complete, enabling all agents to proceed to
  layer~$i+1$.
  Hence each inter-synchronization interval accomplishes
  at most one layer, and at least
  $\Depth_\Phi(\Spec)$ synchronization points are needed.

  \emph{Upper bound.}
  A synchronous protocol in which all agents exchange
  projections at each synchronization point can simulate the
  oracle: the combined projections reconstruct the global
  history, and the strategy selects the next layer's
  commitments.
  At the synchronization point, every agent knows the
  global history (from the exchange), knows that every other
  agent knows it (the exchange was simultaneous), and so on
  at all levels---establishing common knowledge of the
  current state.
  Hence $\Depth_\Phi(\Spec)$ synchronization points suffice.
\end{proof}

The following theorem recovers the Halpern--Moses impossibility
of asynchronous common
knowledge~\cite{halpern1990knowledge} as a consequence of the
determination framework, via an independent proof that uses
only histories, projections, and universal forced dependency.

\begin{theorem}[Asynchronous impossibility via determination]
  \label{thm:async-impossibility}
  Let $\Hist_{\mathsf{async}}$ be an asynchronous history
  class: for every pair of agents $p, q$ and every history
  $H \in \Hist_{\mathsf{async}}$ containing an event $e$ at
  $p$ that is not in the projection of~$q$, there exists
  $H' \in \Hist_{\mathsf{async}}$ with $H \sim_q H'$ in
  which $e$ has not occurred.
  If $\Spec$ has a cross-agent universal forced
  dependency---a
  commitment $\varphi$ at some agent~$p$ and a commitment
  $\psi$ at a distinct agent~$q$ such that $\psi$ universally
  depends on $\varphi$
  (Definition~\ref{def:forced-dependency})---then no
  determination strategy over a commutative basis can
  resolve $\Spec$ over $\Hist_{\mathsf{async}}$.
\end{theorem}

\begin{proof}
  Let $\varphi$ at agent~$p$ and $\psi$ at agent~$q \neq p$
  be a cross-agent universal forced dependency.
  Consider any history $H$ in which $\varphi$ has been
  applied at $p$.
  Since $\varphi$ is an event at $p$ and is not in the
  projection of~$q$, the asynchronous condition guarantees a
  history $H' \in \Hist_{\mathsf{async}}$ with
  $H \sim_q H'$ in which $\varphi$ has not occurred.
  At $H'$, $\psi$ is invalid: by universal forced dependency,
  $\psi$ acts as the identity or violates validity without
  the prior effect of~$\varphi$.
  Since $q$ cannot distinguish $H$ from $H'$, any
  deterministic strategy that applies $\psi$ at $H$ also
  applies it at $H'$, producing an invalid determination.
  Hence no strategy can apply both $\varphi$ and $\psi$
  without a synchronization point between them, and
  $\Spec$ is not resolvable over
  $\Hist_{\mathsf{async}}$.
\end{proof}

\noindent
Theorem~\ref{thm:async-impossibility}---like Halpern-Moses---is a qualitative
threshold: synchronization is needed or not.
Determination depth refines this to a cost measure:
Proposition~\ref{prop:ck-rounds} says at most $k$
synchronization points are needed for depth~$k$ (exactly
$k$ when every layer boundary involves a cross-agent
dependency), and the conservation law
(Theorem~\ref{obs:conservation}) ensures the total
sequential cost cannot be eliminated by enriching the basis.
More broadly, the determination framework applies beyond
the distributed setting---to single-agent offline problems,
combinatorial structures, and autoregressive
generation---making the distributed recovery one
instantiation of a domain-independent measure.

\begin{corollary}[Common knowledge assumptions collapse depth]
  \label{cor:ck-collapse}
  If common knowledge of a property $P$ is assumed as part
  of the model (i.e., $P$ holds at every history in $\Hist$
  and this is common knowledge among all agents), and
  establishing $P$ through synchronization would require
  $j$ layers, then the assumption reduces
  $\Depth_\Phi(\Spec)$ by $j$.
\end{corollary}

\noindent
This unifies the diagnostic observations:
\begin{itemize}[nosep]
  \item \emph{BSP} assumes common knowledge of membership
        (the process set is fixed and globally known),
        saving $1$ layer
        (Section~\ref{sec:bsp-diagnosis}).
  \item In \emph{extensive-form games}, a dominant-strategy
        mechanism makes every node subgame-trivial,
        reducing the strategic depth to~$0$
        (Section~\ref{sec:mechanism}).
  \item \emph{LOCAL} assumes common knowledge of unique
        identifiers, collapsing determination depth to~$1$
        (Section~\ref{sec:local-diagnosis}).
\end{itemize}
In each case, the standard model silently enriches the
commitment basis by assuming common knowledge of a property
whose establishment would otherwise cost determination
layers.

\section{Tightness Proofs}
\label{sec:product-space-proof}
\label{sec:conservation-proof}
\label{sec:conservation-tight-proof}

\paragraph{Constrained generation instantiation.}
The conservation law applies directly to the constrained
generation task (Section~\ref{sec:k-level-task}).
The outcome space is a product
$O = [m]^k$, and the constant-depth basis $\Phi_0$ is the
coordinate basis: each commitment ``fix $v_\ell = v$'' selects
a value for one position, which is computable in $O(1)$
circuit depth.
The dependency DAG $G_{\Phi_0}$ is the chain
$1 \to 2 \to \cdots \to k$, since the feasible values at
position $\ell$ depend on the value chosen at
position $\ell - 1$ (via the successor set
$P_\ell(v_{\ell-1}, \cdot)$).
In Case~2 of the conservation-law proof, this dependency
becomes a literal circuit data path: the sub-circuit computing $v_{\ell_{j+1}}$
must have $v_{\ell_j}$ on its input path, because the
successor sets $P_{\ell_{j+1}}(a, \cdot)$ differ across
predecessor values $a$.
Case~3 cannot arise for the same reason: $v_{\ell_{j+1}}$'s
feasibility depends on $v_{\ell_j}$'s value, so the
successor's exclusion cannot be valid before the
predecessor's.

We first prove the three-way tradeoff, then the conservation
tightness.

\vspace{0.5em}
\noindent
\textbf{Theorem~\ref{thm:three-way}}
(Depth--width--communication tradeoff, restated)\textbf{.}
\textit{For the random $(k,m,s)$-distribution in the
$k$-party model, any $d$-round protocol with width~$w$ and
per-player communication $b_\ell$ bits satisfies
$\log w + \sum_{\ell \in U} b_\ell
  \ge |U| \cdot \log(m/s)$,
where $U$ is the set of uninformed links
($|U| \ge k - d$).}

\begin{proof}[Proof of Theorem~\ref{thm:three-way}]
  Fix a $d$-round protocol with width~$w$ and layer
  assignment $S_1, \ldots, S_d$.
  Let $U$ be the set of uninformed links.
  Index the uninformed links as $\ell_1, \ldots, \ell_t$;
  Observation~\ref{lem:uninformed-link-count} tells us $t \ge k - d$ .

  \emph{Message-decoding lemma.}
  Let $R \subseteq [m]$ be a uniformly random $s$-element
  subset, let $M = M(R)$ be a $b$-bit message (a
  deterministic function of $R$), and let $g$ be any decoder
  that outputs an element $g(M) \in [m]$.
  Then $\Pr[g(M) \in R] \le \min(1,\; 2^b \cdot s/m)$.
  \emph{Proof:}
  Count pairs $(R, \mu)$ with $\mu = M(R)$ and
  $g(\mu) \in R$.
  For each of the at most $2^b$ message values $\mu$, the
  decoder outputs a fixed element $g(\mu)$.
  The number of $s$-subsets containing $g(\mu)$ is
  $\binom{m-1}{s-1}$.
  Hence the total number of good pairs is at most
  $2^b \cdot \binom{m-1}{s-1}$.
  Dividing by the total $\binom{m}{s}$ subsets gives
  $\Pr[g(M) \in R] \le 2^b \cdot
  \binom{m-1}{s-1}/\binom{m}{s} = 2^b \cdot s/m$.
  Taking the minimum with~$1$ gives the claim.

  \emph{Per-link bound.}
  Process the uninformed links in order
  $\ell_1, \ldots, \ell_t$.
  At step~$j$, condition on all constraint functions
  $(P_1, \ldots, P_k)$ except the row
  $P_{\ell_j}(v_{\ell_j - 1}, \cdot)$, on the outcomes
  at links $\ell_1, \ldots, \ell_{j-1}$, and on all
  messages from players other than~$\ell_j$.
  Under this conditioning, $v_{\ell_j - 1}$ is fixed (from
  earlier layers or the same layer), $v_{\ell_j}$ is a
  deterministic function of player~$\ell_j$'s message, and
  $P_{\ell_j}(v_{\ell_j - 1}, \cdot)$ remains a uniformly
  random $s$-subset
  (by Definition~\ref{def:random-ensemble}).
  Applying the message-decoding lemma with $b = b_{\ell_j}$
  (where $v_{\ell_j}$ is determined by player~$\ell_j$'s
  message and the fixed side information under the current
  conditioning),
  the conditional success probability at link $\ell_j$ is at
  most $\min(1,\; 2^{b_{\ell_j}} \cdot s/m)$.

  \emph{Combining.}
  By the chain rule across all uninformed links, the success
  probability of a single candidate is at most
  $\prod_{\ell \in U} \min(1,\; 2^{b_\ell} \cdot s/m)$.
  A union bound over $w$ candidates gives total success
  probability at most
  $w \cdot \prod_{\ell \in U} \min(1,\; 2^{b_\ell} \cdot
  s/m)$.
  For the strategy to succeed with positive probability,
  this upper bound must be at least~$1$.
  Taking logarithms:
  \[
    \log w + \sum_{\ell \in U}
    \min(0,\; b_\ell - \log(m/s))
    \;\ge\; 0.
  \]
  Since $\min(0, x) \le x$,
  \[
    0 \;\le\; \log w + \sum_{\ell \in U}
    \min(0,\; b_\ell - \log(m/s))
    \;\le\; \log w + \sum_{\ell \in U}
    (b_\ell - \log(m/s)),
  \]
  which rearranges to
  $\log w + \sum_{\ell \in U} b_\ell \ge
  |U| \cdot \log(m/s)$.
\end{proof}

\begin{theorem}[Tightness of the conservation law]
  \label{thm:conservation-tight}
  For the $k$-position constrained generation task
  (Definition~\ref{def:k-level-task}) with $m \ge 2s$, every
  determination using $d$ layers satisfies
  $\sum_{i=1}^{d}(1 + c_i) \ge k$, and this bound is achieved
  with equality for every $d \in \{1, \ldots, k\}$.
\end{theorem}

\begin{proof}
  \emph{Lower bound.}
  The dependency DAG for the constrained generation task is
  the chain
  $1 \to 2 \to \cdots \to k$ (each position's feasibility
  depends on the previous value via $P_\ell(v_{\ell-1},
  \cdot)$), which has longest path $k$.
  By Theorem~\ref{obs:conservation},
  $\sum_{i=1}^{d}(1 + c_i) \ge k$ for any $d$-layer
  determination.

  \emph{Upper bound.}
  For any $d \in \{1, \ldots, k\}$, partition the $k$ levels
  into $d$ contiguous blocks $B_1, \ldots, B_d$ of sizes
  $\lfloor k/d \rfloor$ or $\lceil k/d \rceil$.
  In layer~$i$, commit to all positions in $B_i$ by computing
  them sequentially: given the value $v_{\ell-1}$ from the
  previous position (either from an earlier layer or from
  earlier in the same layer's computation), look up any
  $v_\ell \in P_\ell(v_{\ell-1}, \cdot)$.
  Each lookup has circuit depth~$O(1)$ (it is a table scan of
  $P_\ell$), and the $|B_i|$ lookups within layer~$i$ are
  chained, giving $c_i = |B_i| - 1$.
  The total is
  $\sum_{i=1}^{d}(1 + (|B_i| - 1)) = \sum_{i=1}^{d} |B_i|
  = k$.
\end{proof}

\section{Applications to Classical Problems}
\label{sec:applications}
\label{sec:examples}

This section demonstrates the breadth of the determination
framework by applying it to problems across several
domains: stable matching, extensive-form games,
chain-of-thought reasoning, and distributed round complexity
(BSP and LOCAL).
In each case, the framework either recovers a known result
with a new explanation or reveals structure that existing
models do not capture.
Determination depth arises for two reasons across these
examples: \emph{commitment dependency}, where each
commitment creates the sub-problem for the next layer
(constrained generation, extensive-form games); and
\emph{forward validity}, where commitments that are
individually valid conflict when applied in the same layer
(stable matching under the rotation basis).
Each example connects to known results in its domain:
stable matching recovers Garg's parallel
algorithm~\cite{garg2020lattice} and witnesses the
universality of determination depth
(every finite depth is realized by some instance);
extensive-form games measure the game's strategic
depth---the moves where a player must break a tie among
equally-optimal options---connecting
to Selten's trembling-hand
refinement~\cite{selten1975reexamination};
chain-of-thought reasoning decomposes CoT length into
determination depth, computational depth, and architectural
overhead, explaining why longer chains sometimes
degrade~\cite{chen2024illusion,sui2025dontoverthink};
and BSP/LOCAL expose hidden modeling assumptions
(Ameloot et al.'s
non-obliviousness~\cite{ameloot2013relational}) that
silently collapse determination layers.

The examples progress from a purely offline setting (stable
matching) through online single-agent generation
(chain-of-thought) to adversarial online interaction
(extensive-form games) and distributed multi-agent
computation (BSP and LOCAL), illustrating how determination
depth provides exact diagnostics across increasing
environmental complexity.

\subsection{Stable Matching and Determination Universality}
\label{sec:stable-matching}

Stable matching is a canonical relational specification: given
preference lists for two disjoint sets of $n$ agents (encoded
as environment events in an offline history), the
specification
$\Spec_{\mathsf{SM}}$ maps each history to the set of all
stable matchings of the encoded
instance~\cite{gale1962college}.
The number of stable matchings ranges from~$1$ to exponentially
many, and the specification is inherently relational whenever
more than one exists.
The analysis below is in the offline setting: all preference
lists are given and the specification has no extensions.
We show that determination depth for stable matching is exactly
characterized by a classical combinatorial invariant---the height
of the rotation poset---and that stable matching is
\emph{universal} for determination depth: every finite
depth arises as the rotation-poset height of some stable
matching instance.

\subsubsection{Rotations as commitments}

The structural theory of stable matchings is organized around
\emph{rotations}~\cite{irving1986efficient,gusfield1989stable}.
A rotation is a cyclic reassignment of partners that transforms
one stable matching into an adjacent one in the lattice of stable
matchings.
Formally, a rotation $\rho = (a_0, b_0), (a_1, b_1), \ldots,
(a_{r-1}, b_{r-1})$ is a sequence of matched pairs in some
stable matching such that $b_{i+1 \bmod r}$ is the next
partner on $a_i$'s preference list (after $b_i$) with whom
$a_i$ appears in some stable matching (this can be computed
without enumerating all stable
matchings~\cite{irving1986efficient}); applying $\rho$
reassigns each $a_i$ to $b_{i+1 \bmod r}$, producing an
adjacent stable matching in the lattice.

The set of rotations, partially ordered by precedence
($\rho \prec \rho'$ if $\rho$ must be applied before $\rho'$
can be exposed),
forms the \emph{rotation poset} $\Pi$.
The key structural facts
(Irving and Leather~\cite{irving1986efficient},
Gusfield and Irving~\cite{gusfield1989stable}) are:
\begin{enumerate}[label=(\roman*),nosep]
  \item The downsets (closed subsets) of $\Pi$ are in bijection
        with the stable matchings of the instance.
  \item Every finite poset is realizable as the rotation poset
        of some stable matching instance.
\end{enumerate}

We take rotations as the commitment basis $\Phi_{\mathsf{rot}}$.
The commitment for rotation $\rho$ is \emph{ambiguity-sensitive}
(non-pointwise): whether $\rho$ can be applied depends on the
current admissible set, not just on individual outcomes.
Writing $S = \Spec(H)$ for the admissible set at history~$H$:
\[
  \varphi_\rho(S)
  \;\triangleq\;
  \begin{cases}
    \{\,\mu \in S \mid \rho \in \mathsf{ds}(\mu)\,\},
    & \text{if $\rho$ is \emph{exposed} in $S$}\\
    & \text{(all predecessors of $\rho$ in $\Pi$}\\
    & \text{are in $\mathsf{ds}(\mu)$ for every $\mu \in S$),}\\
    S, & \text{otherwise,}
  \end{cases}
\]
where $\mathsf{ds}(\mu)$ denotes the downset of $\Pi$
corresponding to matching~$\mu$
(by fact~(i) above, each stable matching corresponds to a
unique downset).
The commitment is irrevocable and satisfies shrinkage
(it filters the admissible set rather than transforming
matchings; in the determination framework, ``applying a
rotation'' means retaining only matchings consistent with
that rotation having been applied).
Any determination that applies rotations in a valid
topological order preserves feasibility.

\subsubsection{Determination depth equals rotation poset height}

\begin{proposition}[Determination depth of stable matching]
  \label{thm:sm-depth}
  For any stable matching instance with rotation poset $\Pi$,
  \[
    \Depth_{\Phi_{\mathsf{rot}}}(\Spec_{\mathsf{SM}})
    \;=\;
    \mathrm{height}(\Pi),
  \]
  where $\mathrm{height}(\Pi)$ is the length of the longest
  chain in $\Pi$.
\end{proposition}

\begin{proof}
  \emph{Upper bound.}
  Partition the rotations of $\Pi$ into layers by a longest-path
  layering: layer~$i$ contains all rotations whose longest chain
  of predecessors has length~$i$.
  After layers $1, \ldots, i\!-\!1$ have been applied, every
  rotation in layer~$i$ is exposed: by the downset--matching
  bijection~(i), applying a rotation $\rho'$ retains only
  matchings whose downsets contain $\rho'$, so after all
  predecessors of a layer-$i$ rotation have been applied,
  every surviving matching's downset contains them.
  Within layer~$i$, the rotations are pairwise incomparable in
  $\Pi$; since all are simultaneously exposed, applying any
  subset does not affect the exposure status of the others.
  (Exposure of $\rho'$ requires all predecessors of $\rho'$ to
  be in $\mathsf{ds}(\mu)$ for every surviving $\mu \in S$.
  Applying an incomparable $\rho$ only shrinks $S$, which can
  only make this condition easier to satisfy, not harder.)
  Hence the commitments within each layer commute, and the
  number of layers equals $\mathrm{height}(\Pi)$.

  \emph{Lower bound.}
  Let $\rho_1 \prec \rho_2 \prec \cdots \prec \rho_h$ be a
  longest chain in $\Pi$.
  Each $\varphi_{\rho_{i+1}}$ depends on $\varphi_{\rho_i}$
  (Definition~\ref{def:forced-dependency}): before $\rho_i$
  is applied, the admissible set contains matchings whose
  downsets do not include $\rho_i$, so $\rho_{i+1}$ is not
  exposed and $\varphi_{\rho_{i+1}}$ acts as the identity.
  This gives a dependency chain of length~$h$, so
  $\Depth_{\Phi_{\mathsf{rot}}}(\Spec_{\mathsf{SM}}) \ge h$
  by Proposition~\ref{thm:dependency-chain-lb}.
\end{proof}

\begin{remark}[Universality]
  \label{cor:sm-universality}
  By the Irving--Leather realization
  theorem~\cite{gusfield1989stable}, every finite poset arises
  as a rotation poset.
  Since determination depth equals poset height
  (Proposition~\ref{thm:sm-depth}), stable matching realizes
  every possible determination depth: for every $k$, there
  exists an instance with determination depth exactly~$k$.
\end{remark}

\paragraph{Depth-optimality of strategies.}
The Gale--Shapley algorithm~\cite{gale1962college} finds a
stable matching in $O(n^2)$ sequential steps, but it is not
depth-optimal: it uses $O(n^2)$ sequential steps even when the
rotation poset has small height.
A depth-optimal strategy---applying all exposed rotations
simultaneously at each layer---achieves depth equal to the
poset height, potentially much less than $O(n^2)$.
This is precisely the parallel algorithm derived independently
by Garg's lattice-linear predicate
framework~\cite{garg2020lattice}, which the determination
framework recovers as a consequence of the depth
characterization.

\subsection{Strategic Depth in Extensive-Form Games}
\label{sec:mechanism}

How much arbitrary choice does a game force upon a player?
Game-tree depth overcounts: at forced moves, where only one
option is optimal, the player faces a purely computational
burden, not a relational one.
We use determination depth to measure a game's
\emph{strategic depth}: the number of moves where the
player must break a tie among multiple equally-optimal
options---an irreducible relational cost that no amount of
computation can eliminate.

\paragraph{Specification.}
Consider a two-player extensive-form game with perfect
information.
Player~1 (the determiner) and Player~2 (the environment)
alternate moves in a game tree~$T$; this is an online
specification, since player~2's moves are environment events
that extend the history between player~1's commitments.
The outcome set $O$ is the set of all root-to-leaf paths
(complete plays), and the initial admissible set is
$\Spec(H_0) = O$.
We restrict the admissible set to plays consistent with
\emph{subgame-perfect equilibrium}
(SPE)~\cite{selten1965spieltheoretische}: at each node,
the acting player's move must be optimal given optimal play
in all subsequent subgames.
Formally, the admissible set at history~$H$ is the set of
all complete plays extending~$H$ that arise under
\emph{some} subgame-perfect equilibrium of the full game.
The specification is relational whenever multiple SPE plays
pass through the current history.

\paragraph{Commitment basis.}
At each player-1 node~$v$, the commitment ``choose
child~$c$'' excludes all plays not passing through~$c$:
a pointwise filter.
Player-2 moves are environment events---outside the
determiner's control, regardless of player-2's equilibrium
rationality.
Player-1 moves at \emph{subgame-trivial} nodes (where the
SPE prescribes a unique move) are also effectively
environment events: the move is uniquely determined, so no
choice is involved.
Only moves at \emph{subgame-non-trivial} nodes---where
multiple children lead to plays with the same SPE value for
player~1---are genuine commitments requiring an arbitrary
choice.

\begin{definition}[Subgame-non-trivial node]
  \label{def:nontrivial-node}
  A player-1 node~$v$ is \emph{subgame-non-trivial} if at
  least two children of~$v$ are each consistent with some
  (possibly different) subgame-perfect equilibrium of the
  full game.
  It is \emph{subgame-trivial} if exactly one child is
  SPE-consistent.
\end{definition}

\begin{proposition}[Strategic depth of extensive-form games]
  \label{prop:strategic-depth}
  For a two-player extensive-form game with perfect
  information, the determination depth from player~1's
  perspective (under the SPE-restricted specification) is at
  most the maximum number of subgame-non-trivial player-1
  nodes on any root-to-leaf path.
  The bound is tight whenever, between each consecutive pair
  of such nodes, player~2 has at least one
  subgame-non-trivial node.
\end{proposition}

\begin{proof}
  \emph{Upper bound.}
  A strategy that auto-plays the unique SPE move at
  subgame-trivial nodes and commits at subgame-non-trivial
  nodes uses one layer per non-trivial node along any
  realized play.
  Forced moves (at trivial nodes) and player-2 responses
  occur between commitment layers as environment events,
  contributing no determination depth.

  \emph{Lower bound (under the tightness condition).}
  Consider two consecutive subgame-non-trivial player-1
  nodes $u$ and $v$ on a root-to-leaf path, with $v$
  deeper than~$u$.
  Node~$v$ has SPE-consistent children~$c_1, c_2$, so
  there exist SPEs $\sigma_1, \sigma_2$ of the full game
  inducing plays through~$c_1$ and~$c_2$ respectively.
  Between $u$ and $v$, player-2 moves and trivial
  player-1 moves occur as environment events.
  The SPEs $\sigma_1$ and $\sigma_2$ may prescribe
  different player-2 moves at nodes between $u$ and~$v$,
  so the admissible set at~$v$---which plays remain
  SPE-consistent---depends on which environment events
  materialize after the commitment at~$u$.
  Under one sequence of player-2 moves, only plays
  through~$c_1$ may remain admissible at~$v$; under
  another, only plays through~$c_2$.
  Since these environment events have not occurred when
  the commitment at~$u$ is made, no commitment at~$u$
  can validly determine the choice at~$v$: any fixed
  choice risks selecting a child whose plays are
  inadmissible under some SPE-consistent continuation.
  The choice at~$v$ is therefore a forced dependency on
  the history up to~$v$.
  Chaining these forced dependencies across all
  non-trivial nodes on the path gives a dependency chain
  of length equal to their count
  (Proposition~\ref{thm:dependency-chain-lb}).
\end{proof}

\noindent
The result decomposes game-tree depth into two components:
\emph{strategic depth} (the non-trivial nodes, where the
player must break a tie among equally-optimal moves) and
\emph{forced depth} (the trivial nodes, where the optimal
move is uniquely determined and no choice is involved).
This decomposition is distinct from classical game-tree
depth and alternating Turing machine quantifier depth,
both of which count \emph{all} player-1 moves---including
forced ones that involve no relational choice.
The distinction is analogous to the BSP diagnostic
(Section~\ref{sec:bsp-diagnosis}): just as BSP overcharges
by counting communication rounds that resolve no ambiguity,
game-tree depth overcharges by counting moves that present
no relational choice.
The analysis is restricted to finite perfect-information
games (Zermelo's setting~\cite{zermelo1913anwendung}); extending to imperfect
information requires handling information sets and is left
open.

\begin{remark}[Error amplification under bounded rationality]
  \label{rem:error-amplification}
  Strategic depth has consequences beyond the relational
  burden itself.
  If a bounded player can compute SPE values but trembles
  when breaking ties---erring at each non-trivial node with
  independent probability~$p$---the probability of perfect
  play is $(1-p)^d$ ($d$ = strategic depth, not game-tree
  depth~$k$).
  Forced moves present no relational risk; computational
  failures at forced moves are an orthogonal axis of bounded
  rationality.
  Strategic depth is per-player and need not be symmetric:
  one player can often influence the other's strategic depth
  by choosing which subtree to enter.
  In chess, this is ``playing for complications'': steering
  toward positions where the opponent faces more
  tie-breaking choices.
  Balancing payoff maximization against error
  amplification is an open question that connects to
  Selten's trembling-hand
  refinement~\cite{selten1975reexamination}, with
  influence over determination depth as a new degree of
  strategic freedom.
\end{remark}

\begin{remark}[Mechanism design as basis enrichment]
  \label{rem:mechanism-basis}
  Beyond extensive-form games, the determination framework
  applies to mechanism design.
  A mechanism (e.g., an auction format) can be viewed as
  enriching the commitment basis: a direct-revelation
  mechanism allows a player to submit a full strategy as a
  single commitment, collapsing multiple sequential choices
  into one layer.
  A dominant-strategy mechanism (e.g., a second-price
  auction~\cite{vickrey1961counterspeculation}) makes every
  node subgame-trivial---each player's optimal action is
  independent of others'---collapsing strategic depth
  to~$0$.
  An indirect mechanism (e.g., an ascending
  auction~\cite{milgrom2000putting}) trades strategic depth
  for communication: more rounds of bidding, but less
  information per round.
  Whether indirect mechanisms obey a conservation law
  analogous to Theorem~\ref{obs:conservation}---trading
  determination depth against per-round communication and
  outcome-set width---is an open problem.
\end{remark}
\subsection{Decomposing Chain-of-Thought Length}
\label{sec:cot-appendix}
\label{sec:cot-decomposition}

An autoregressive transformer generates tokens sequentially,
each an irrevocable commitment.
Chain-of-thought (CoT) prompting increases the sequential layers
available.
Recent work shows CoT provides exponential
advantages~\cite{feng2023towards,li2024chain,mirtaheri2025letmethink},
yet longer chains sometimes
degrade~\cite{chen2024illusion,sui2025dontoverthink}.

We use the determination framework to decompose CoT length
into formally independent components.
The decomposition is not a theorem about any specific
transformer architecture; it is a semantic lower bound that
applies to any autoregressive procedure in which each token
is an irrevocable commitment and the task is relational
(multiple valid continuations exist).
The constrained generation task
(Section~\ref{sec:exp-separation}) serves as the witness
because it isolates determination depth from computational
depth: each commitment is a constant-time table lookup, so
the sequential cost is entirely semantic.
The total CoT length $T$ is lower-bounded by three
independent quantities:
(i)~\emph{determination depth}~$d$---non-commuting commitment
layers, witnessed by the constrained generation task
(Section~\ref{sec:exp-separation});
(ii)~\emph{computational depth}~$c$---inherent per-step
sequentiality from bounded transformer depth
(TC$^0$-like~\cite{feng2023towards}), witnessed by graph
connectivity~\cite{mirtaheri2025letmethink}; and
(iii)~\emph{architectural overhead}---tokens spent on
information management (context summarization, backtracking,
re-derivation), reducible by changing the architecture
without changing the task.

\begin{proposition}[CoT lower bound from determination depth]
  \label{thm:cot-lower-bound}
  CoT length $T \ge \max(d, c)$: the bound $T \ge d$ follows
  from Theorem~\ref{thm:exp-separation} and $T \ge c$
  from~\cite{feng2023towards,mirtaheri2025letmethink}.
  The two bounds are independent: there exist tasks with $d = k$,
  $c = O(1)$ (Theorem~\ref{thm:exp-separation}), and tasks with
  $d = 0$,
  $c = \omega(\!\log n)$~\cite{mirtaheri2025letmethink}.
  In both cases, fewer than $\max(d,c)$ layers requires
  exponential parallel width.
\end{proposition}

The decomposition explains the apparent conflict.
Tasks where CoT provides exponential advantage
are determination-bound or computation-bound: their
sequential cost is intrinsic, and more layers directly reduce
it.
Tasks where longer chains
degrade~\cite{chen2024illusion,sui2025dontoverthink} are
architecture-bound: $d$ and $c$ are small, so additional
tokens are spent on overhead (context management,
backtracking, re-derivation) that introduces errors without
reducing the bottleneck.
For determination-bound tasks, the conditional-spread
parameter $\gamma$ governs how much distributional
knowledge can substitute for chain length: a model facing a
$\gamma$-spread distribution cannot benefit from additional
training, while a model facing a more structured distribution
can trade prediction quality for fewer layers
(Remark~\ref{sec:approx-determination}).

\subsection{Diagnosing Distributed Round Complexity}
\label{sec:bsp-diagnosis}
\label{sec:example-bsp-collapse}

BSP~\cite{valiant1990bridging} organizes parallel computation
into synchronous rounds separated by barriers.
We show that BSP round complexity and determination depth
disagree in both directions, and that the disagreement is
precisely explained by the gap between BSP's implicit
commitment basis and the atomic basis that
serves as our intrinsic reference
(Section~\ref{sec:intrinsic-depth}).

\paragraph{The atomic basis as intrinsic reference.}
Under the atomic basis in the online setting, determination
depth measures the
irreducible commitment structure of a specification---the
minimum number of non-commuting layers forced by forward
validity constraints, with no assumptions about process identity,
membership, or communication primitives.
This is the ``bare'' complexity of the specification.
BSP departs from this reference in two ways: it
charges for communication (adding rounds where no commitment
occurs) and it implicitly enriches the
basis (collapsing layers that the atomic basis exposes).

\paragraph{Overcharging: rounds without commitments.}
BSP charges one round for any operation requiring
communication, whether or not the communication resolves
semantic ambiguity.
A single-valued specification has determination depth~$0$
under any basis---the output is uniquely determined, so no
commitment is needed.
Relational join is the canonical example:
parallel implementations that require no synchronization
barriers have been known since the early
1990s~\cite{wilschut1991parallel}, yet BSP assigns one round
to the data shuffle because BSP rounds account for
communication uniformly.
The unnecessary barrier has practical consequences
(e.g., straggler sensitivity and inability to pipeline),
which the determination framework diagnoses as overcharging
for a depth-$0$ task.

\paragraph{Undercharging: basis enrichment collapses depth.}
BSP implicitly enriches the atomic basis by assuming
\emph{fixed, static membership}: the process set is finite,
globally known, and does not change during execution.
In a real system, establishing membership requires a prior
commitment---a discovery protocol, a configuration step, or
a leader's decision about who participates---that BSP treats
as given.
This assumption hides exactly one unit of determination depth.
Under the atomic basis, resolving membership costs at least
one layer: a membership commitment (which processes
participate?) must precede a value commitment (what does each
process output?), and the two do not commute, giving
depth~$2$.
BSP absorbs the membership commitment into its model
assumptions, leaving only the value commitment visible---one
BSP round for a problem with intrinsic depth~$2$.
By the conservation law
(Theorem~\ref{obs:conservation}), the total sequential
depth is unchanged; BSP hides the membership layer
rather than eliminating it
(Corollary~\ref{cor:ck-collapse} in
Section~\ref{sec:distributed-ck}, with $j = 1$).

The consequences are sharp.
With static membership, universally quantified conditions
(e.g., verifying that all processes have completed a round)
reduce to finite conjunctions over a known process
set~\cite{immerman1986relational}, and all remaining
non-monotonic reasoning can proceed without distributed
coordination by the CALM theorem~\cite{ameloot2013relational}.
With \emph{dynamic} membership, Ameloot's preconditions fail:
universal quantification over participants requires a
prior determination of who participates, and each change in
membership costs an additional layer.

\paragraph{LOCAL model.}
\label{sec:local-diagnosis}
The same diagnostic applies to the LOCAL model of
distributed graph computation.

Consider $(\Delta+1)$-coloring of a graph $G = (V, E)$ on $n$
nodes.
The specification is relational: many valid colorings exist.
The synchronous LOCAL model~\cite{linial1992locality} provides
synchronized rounds, simultaneous neighbor-state revelation,
and unique node identifiers.
Under this basis, $(\Delta+1)$-coloring has determination
depth~$1$: deterministic symmetry-breaking algorithms (e.g.,
Cole--Vishkin~\cite{cole1986deterministic}) compute a proper
coloring as a function of the ID-labeled neighborhood in
$\Theta(\log^* n)$ rounds.
The round complexity is entirely computational---it measures
the cost of symmetry breaking, not of semantic commitment.
Common knowledge of unique identifiers collapses the
determination structure entirely
(Corollary~\ref{cor:ck-collapse}).

IDs alone account for the collapse---synchronization and
within-layer communication are not needed.
Each node waits (asynchronously) for its lower-ID neighbors
to commit, then takes the smallest available color.
The ID ordering prevents circular dependencies, so no
barrier or broadcast is needed.
In the language of Ameloot et
al.~\cite{ameloot2013relational}, both static membership
(Section~\ref{sec:bsp-diagnosis}) and unique identifiers
are instances of \emph{non-obliviousness}: shared knowledge
that eliminates the need for distributed coordination.
Establishing common knowledge of identifiers lets all
remaining non-monotonic reasoning proceed locally, but the
conservation law
(Theorem~\ref{obs:conservation}) ensures that the sequential
cost persists as local computation within layers.
For coloring, IDs happen to collapse both costs (greedy
coloring along the ID-induced orientation is a single-layer
local computation), but this is a special property of the
coloring specification, not a general consequence of having
IDs.

\begin{remark}[Randomized strategies]
  \label{rem:randomized-local}
  The analysis above applies to deterministic strategies.
  With randomization, the Lov\'{a}sz Local Lemma yields
  $O(1)$-round distributed algorithms for
  $(\Delta\!+\!1)$-coloring with high
  probability~\cite{moser2010constructive}, collapsing
  determination depth to $O(1)$ even without IDs.
  Randomization thus provides a mechanism for collapsing
  determination depth by breaking symmetry probabilistically.
  Characterizing which specifications admit randomized
  depth collapse is an open question.
\end{remark}

\section{Metacomplexity Proofs}
\label{sec:ic}
\label{sec:metacomplexity}

How hard is it to compute the determination depth of a given
specification?
This section shows that the answer ranges from NP-hard to
PSPACE-complete.
In the offline setting, computing determination depth under a
pointwise basis is already NP-hard (via decision tree
synthesis).
In the online setting, the alternation between the
determiner's commitments and the environment's history
extensions produces quantifier alternation:
``is depth $\le k$?'' is
$\Sigma_{2k}^P$-complete for each fixed~$k$ and
PSPACE-complete for unbounded~$k$, so the polynomial
hierarchy is precisely the hierarchy of determination
depths.

\begin{proposition}[Metacomplexity of determination depth]
  \label{thm:metacomplexity}
  Computing determination depth is NP-hard in the offline
  setting: for specifications arising from decision tree
  synthesis
  (``given a truth table, output an optimal decision tree''),
  determination depth equals the minimum decision tree depth,
  which is NP-hard to
  compute~\cite{hyafil1976constructing}.
\end{proposition}

\begin{proof}
  \emph{Decision tree synthesis.}
  The environment events encode a truth table
  $f : \{0,1\}^n \to \{0,1\}$ (an offline setting with no
  further extensions), and the outcome set $O$ is the family
  of all decision trees on $n$~variables.
  The specification maps each complete
  history---encoding a truth table---to the set of decision
  trees that compute it:
  $\Spec(H) = \{\,T : T \text{ computes } f_H\,\}$, where
  $f_H$ is the truth table encoded in~$H$.
  The commitment basis consists of pointwise filters
  ``test variable $x_i$ at the current node'': each such
  commitment restricts the admissible set to decision trees
  that test $x_i$ at that node.
  Because each node of a decision tree tests exactly one
  variable, two test commitments for different variables at
  the same node have no common tree in their
  intersection---the admissible set becomes empty---so any
  valid determination must select exactly one variable per
  node.
  A determination of depth $d$ corresponds to a decision tree
  of depth $d$: each layer selects a variable to test at each
  node of that layer, and the two test outcomes branch into
  sub-problems resolved by subsequent layers.
  Any depth-$d$ determination yields a depth-$d$ tree, and
  vice versa, so determination depth equals minimum decision
  tree depth, which is NP-hard to
  compute~\cite{hyafil1976constructing}.
  Depth here reflects the tree structure: a depth-$d$
  decision tree has $d$ levels of nodes, and each level
  is one round of variable-test commitments.
\end{proof}

The NP-hardness result above is an offline result: the
specification is fully given and there is no adversarial
environment.
In the online setting, an adversarial environment adds
universal quantification---the environment fixes some
variables, the determiner fixes others---and the
metacomplexity rises through the full polynomial
hierarchy.

To state the complexity result precisely, consider the
following class of online specifications.
Given a Boolean formula $\theta(x_1, \ldots, x_n)$
(encoded by a polynomial-size circuit), define a
specification with outcome set
$O = \{0,1\}^n$ (all variable assignments) and initial
admissible set $\Spec(H_0) = O$ (every
assignment is initially admissible).
The commitment basis consists of pointwise filters
``set variable $x_i = b$'': each excludes all assignments
with $x_i \neq b$, shrinking the admissible set without
resolving it completely.
The game has $n$ rounds.
A \emph{level assignment} partitions the $n$ rounds between
two players: at each \emph{determiner} round the determiner
fixes one variable (an existential choice); at each
\emph{environment} round the environment fixes one variable
adversarially (a universal choice).
After all $n$ rounds every variable is fixed and the
admissible set is a singleton.
The \emph{determination depth} is the number~$k$ of
determiner rounds; the metacomplexity question is: given
$\theta$ and $k$, does there exist a level assignment with
$k$ determiner rounds and a strategy for the determiner
that guarantees the final assignment
satisfies~$\theta$, regardless of the environment's choices?
The input size is polynomial in~$n$ (the circuit encoding
of~$\theta$) even though $|O| = 2^n$.

\begin{theorem}[Determination depth captures the polynomial
  hierarchy]
  \label{thm:ph-characterization}
  For the class of specifications above:
  \begin{enumerate}[label=(\roman*),nosep]
    \item For each fixed $k$, the problem ``is determination
          depth $\le k$?'' is $\Sigma_{2k}^P$-complete.
    \item For unbounded $k$ (given as input), the problem is
          PSPACE-complete.
  \end{enumerate}
\end{theorem}

\noindent
This is the standard setup for QBF and PSPACE-complete game
evaluation.

\begin{proof}
  \emph{Upper bound.}
  The determiner guesses which $k$ of the $n$ rounds to
  control.
  Each determiner choice is an existential quantifier; each
  environment choice is a universal quantifier.
  After all $n$ rounds, every variable is fixed and
  satisfaction of $\theta$ is checkable in polynomial time
  by evaluating the circuit.
  Consecutive rounds controlled by the same player collapse
  into a single quantifier block, so the quantifier pattern
  has at most~$2k$ alternations, placing the problem in
  $\Sigma_{2k}^P$.
  When $k$ is part of the input, the number of alternations
  $2k$ is at most $2n$; since the circuit encoding
  of~$\theta$ has at least $n$ input wires, $2n$ is linear
  in the input size.
  A single alternating polynomial-time machine can therefore
  handle every~$k$, placing the problem in
  APTIME $=$ PSPACE~\cite{sipser1983hierarchy}.

  \emph{Hardness.}
  Reduce from $\Sigma_{2k}$-QBF for fixed~$k$, or from
  TQBF for unbounded~$k$.
  Given a quantified Boolean formula\\
  $\exists y_1 \forall x_1 \cdots
  \exists y_k \forall x_k.\;
  \theta(x, y)$, construct a specification with $2k$
  variables and outcome set $O = \{0,1\}^{2k}$.
  Odd-numbered rounds ($1, 3, \ldots, 2k\!-\!1$) are
  controlled by the determiner, fixing the existential
  variables $y_1, \ldots, y_k$; even-numbered rounds
  ($2, 4, \ldots, 2k$) are controlled by the environment,
  fixing the universal variables $x_1, \ldots, x_k$.
  The quantifier alternation of the QBF maps directly onto
  the round structure: each $\exists$ becomes a determiner
  round, each $\forall$ an environment round.
  The initial admissible set is $\Spec(H_0) = O$ (all
  assignments); each round's commitment narrows it by fixing
  one variable.
  After all $2k$ rounds the assignment is fully determined;
  the determiner's strategy succeeds iff the resulting
  assignment satisfies~$\theta$.

  \emph{Correctness} (the QBF is true iff the determiner has
  a depth-$k$ strategy).
  If the QBF is true, the existential player has a winning
  strategy: a choice of each $y_i$ (possibly depending on
  $x_1, \ldots, x_{i-1}$) such that $\theta$ is satisfied
  for every universal assignment.
  The determiner plays this strategy at the odd rounds,
  using $k$ layers (one per existential variable), so
  depth~$\le k$.
  Conversely, any depth-$k$ determiner strategy is an
  adaptive assignment to the existential variables that
  satisfies $\theta$ against every universal response---exactly
  a witness that the QBF is true.
\end{proof}

\noindent
The polynomial hierarchy is thus precisely the hierarchy of
metacomplexity for determination depth: deciding whether a
specification has depth $\le k$ is $\Sigma_{2k}^P$-complete,
and PSPACE-complete for unbounded~$k$.
In the offline setting
(Proposition~\ref{thm:metacomplexity}), the metacomplexity
is NP-hard but does not capture the full hierarchy.

\section{Related Work}
\label{sec:related-work}

This section positions determination depth relative to
existing complexity measures that involve sequential
staging.
These measures treat staging either as an operational
resource (rounds, adaptivity, oracle calls) or as a
syntactic discipline (stratification, fixpoint nesting);
determination depth differs in that it measures the
semantic cost of irrevocable commitment, independent of any
particular model or language.

\paragraph{Parallelism, depth, and inherent sequentiality.}
Depth as a complexity measure has a long history in models of parallel
computation.
Circuit complexity distinguishes size from depth, isolating irreducible
sequential structure even when unbounded parallelism is available
\cite{pippenger1980comparisons,barrington1990regular}.
Circuit depth measures the inherent sequentiality of \emph{evaluating} a
function: data dependencies force some gates to wait for others.
Determination depth measures the inherent sequentiality of
\emph{choosing} an outcome from a relation: commitment dependencies
force some decisions to wait for others.
The two are distinct: a function has determination depth~$0$
regardless of its circuit depth, while a relation can have
high determination depth even when every individual
commitment is a constant-depth circuit
(Section~\ref{sec:exp-separation}).
In the online setting, forward validity constraints add a
further source of determination depth that circuit models
do not capture at all.
PRAM and BSP models similarly separate local computation from global
synchronization, charging depth to barriers or rounds
\cite{borodin1982time,valiant1990bridging}.
In the online (distributed) setting, these models can both undercharge
(when a round boundary hides an irrevocable commitment that the model
treats as primitive) and overcharge (when communication structure is
conflated with semantic resolution);
Section~\ref{sec:bsp-diagnosis} develops this diagnostic in detail.
More broadly, classical hierarchy theorems establish strict
separations based on time, space, or alternation
depth~\cite{sipser1983hierarchy}; our hierarchies arise
from semantic dependency rather than resource bounds.

\paragraph{Trace monoids and partial commutation.}
The layered normal form (Section~\ref{sec:normal-forms})
resembles the Foata normal form of Mazurkiewicz trace
theory~\cite{mazurkiewicz1977concurrent}: a canonical
factorization of a word in a partially commutative monoid into
maximal independent steps.
The resemblance is structural but the theories diverge in three
ways that produce qualitatively different phenomena.
First, in classical trace theory the independence relation is
\emph{fixed}: two letters either commute or they do not,
regardless of context.
In determination theory, commutation is \emph{dynamic}: whether
two commitments commute depends on the current refined
specification, and applying one commitment can create or destroy
independence among the remaining ones.
This dynamic commutation arises for different reasons in the two
settings: in the online setting, forward validity constraints
(a commitment that is safe alone may become unsafe after another
commitment narrows the admissible set at some extension) create
and destroy independence; in the offline setting, commitment
dependency
(committing to one position determines which choices remain
feasible at the next) produces the same effect.
Second, this state-dependence makes the analogue of the Foata
normal form non-unique---different layering choices lead to
different refined specifications and potentially different
minimum determination depths---and computing the minimum height becomes
NP-hard (Proposition~\ref{thm:metacomplexity}), in contrast to
the polynomial-time computability of static Foata height.
Third, the main results of this paper---the exponential
depth--width separation (offline), the oracle characterization
(online), and the conservation law (both settings)---have no
analogues in classical trace theory.
They arise from the semantic content of commitments (feasibility,
resolution, forward validity) rather than from the algebraic
structure of partial commutation alone.
In short, trace theory provides the algebraic skeleton; the
semantic content of determination fills it with phenomena that
the skeleton alone cannot express.

\paragraph{Oracle models and adaptivity.}
The use of a disclosure oracle to characterize depth parallels classical
distinctions between adaptive and non-adaptive computation.
Decision tree complexity, oracle Turing machines, and communication complexity
all exhibit hierarchies based on the number of adaptive rounds
\cite{karp1980turing,papadimitriou1984communication,nisan1991crews}.
Our oracle characterization
(Proposition~\ref{thm:oracle-depth-body}) is specific to the
online setting under a commutative basis: oracle invocations
witness irreducible \emph{semantic} dependency---points at which
the strategy must observe the evolving history before committing
further---rather than query adaptivity or information flow.
Local computation between oracle calls is unrestricted, emphasizing
that determination depth cannot be collapsed by parallelism or
control flow alone.
In the offline setting, a single oracle call reveals the complete
input and the oracle count collapses to one; the sequential cost
reappears as computational depth within layers (the conservation
law, Theorem~\ref{obs:conservation}).
In this sense, determination depth and adaptivity depth both
measure staged dependency, but the dependency has a different
character.
An adaptive query reveals information about a fixed,
predetermined answer; a determination commitment \emph{creates}
the answer by irrevocably excluding alternatives.
Adaptivity depth measures how many times a strategy must look;
determination depth measures how many times it must choose.
The two are independent (Observation~\ref{prop:pointer-chasing},
Section~\ref{sec:exp-separation}): pointer chasing at sparsity $s = 1$
has high adaptivity depth but determination depth zero (the
answer is unique); constrained generation in the offline
setting has determination depth~$k$ but adaptivity depth zero
(the entire input is given).

\paragraph{Round hierarchies in communication complexity.}
The closest technical antecedent to our exponential separation
is the Nisan--Wigderson round hierarchy for pointer
chasing~\cite{nisan1993rounds}; our constrained generation
task extends the same chain structure from functions to
relations
(Observation~\ref{prop:pointer-chasing},
Theorem~\ref{thm:three-way}).
The round-elimination technique---reducing a $k$-round
protocol to a $(k\!-\!1)$-round protocol with a
communication blowup---is the standard tool for proving
round lower bounds in communication
complexity~\cite{miltersen1998data,sen2018rounds}.
Mao, Yang, and Zhang~\cite{mao2025gadgetless} recently obtained tight
pointer-chasing bounds via gadgetless lifting, bypassing
round elimination entirely.
Both techniques---round elimination and gadgetless
lifting---operate in the functional regime ($s = 1$), where
hardness is informational: the answer is unique but
distributed across players.
Our lower bound targets the relational regime ($s \ge 2$),
where the answer is not unique and hardness arises from the
cost of committing to one answer among many, not from
distributing information about a fixed answer.
Our exponential separation is an offline result (the entire
constraint chain is given to a centralized machine), while
pointer chasing is an online communication problem (the chain
is distributed across players who communicate in rounds).
The three-way tradeoff (Theorem~\ref{thm:three-way}) unifies
both: it interpolates between the offline width bound
(setting communication to zero) and the online communication
bound (setting width to one).
Our lower-bound technique (conditional spread and union bounds
over uninformed links) exploits the relational structure of
the specification rather than information-theoretic arguments
about message content.
The key qualitative difference is that round-elimination
lower bounds are \emph{polynomial} in the communication
parameter, while our depth--width tradeoff is
\emph{exponential}---and the transition occurs precisely at
the boundary between functional and relational
specifications ($s = 1$ vs.\ $s \ge 2$).

\paragraph{Logic, stratification, and fixpoints in databases.}
Layered semantics appear in database theory through stratified
negation, well-founded semantics, and alternating
fixpoints~\cite{vangelder1991wellfounded}; in descriptive
complexity, fixpoint nesting and alternation depth classify
expressive power~\cite{immerman1999descriptive}.
These layerings are properties of a \emph{program under a
chosen semantics}.
Determination depth, by contrast, is a property of the
\emph{specification} (the mapping from histories to admissible
outcome sets), independent of any particular logical formalism
or semantics: two programs that define the same specification
have the same determination depth, even if they differ in
strata count, fixpoint nesting, or use of negation.
The distinction matters because different negation semantics
applied to the \emph{same} program can yield different
specifications and hence different determination depths.
Each semantics induces a different commitment
basis---stratified semantics requires a linear chain of
sealing commitments, well-founded semantics an alternating
sequence, stable semantics a shallow branching among
incompatible determinations---a distinction strata
counts alone cannot express.

\paragraph{Leaf languages.}
The leaf language framework of Bovet, Crescenzi, and
Silvestri~\cite{bovet1992leaf} characterizes complexity classes
by the language-theoretic complexity of the string of outcomes
at the leaves of a nondeterministic computation tree: NP
corresponds to testing membership in $\{0,1\}^*1\{0,1\}^*$,
PSPACE to a regular language recognizable with $O(\log n)$
memory, and so on.
Both frameworks use the structure of a tree of possibilities
to classify problems, but the objects and questions differ.
Leaf languages classify the \emph{acceptance condition} applied
to a fixed nondeterministic computation; determination depth
classifies the \emph{commitment structure} required to resolve
a relational specification.
In particular, leaf languages operate on a single computation
tree whose branching is fixed by the machine, while
determination depth measures layered commitment across
histories whose extensions are chosen by an adversarial
environment---a distinction that is specific to the online
setting.
The PH characterization
(Theorem~\ref{thm:ph-characterization-body}) recovers the
same hierarchy that leaf languages capture, but through
alternation of commitment and environment moves in the online
game rather than through the complexity of a leaf-string
acceptance condition.
In the offline setting, the adversarial environment disappears
and the exponential separation
(Theorem~\ref{thm:exp-separation}) provides a different route
to determination depth hierarchies, via commitment dependency
rather than alternation.
More broadly, the leaf language program provided an elegant
\emph{classification} of existing complexity classes but did
not yield new separations or lower bounds.
Our framework is aimed at a different target: not reclassifying
known classes, but identifying an axis of complexity distinct
from computation
(Section~\ref{sec:exp-separation}),
producing new tradeoffs
(Theorems~\ref{thm:exp-separation}--\ref{thm:three-way}),
and diagnosing existing models
(Sections~\ref{sec:bsp-diagnosis}--\ref{sec:local-diagnosis}).

\paragraph{Coordination and monotonicity.}
The Coordination Criterion and the CALM theorem relate monotonicity to the absence
of coordination requirements in distributed systems
\cite{hellerstein2026coordinationcriterion,ameloot2013relational}.
These are inherently online results: coordination cost arises from
the need to commit before the full history is known.
In our framework, monotone (future-monotone) specifications
correspond to the \emph{depth-zero} fragment, in which
no irrevocable commitment is needed at any prefix.
Deeper specifications require staged commitments in the online
setting, even in the presence of powerful coordination
primitives; the distributed extension
(Section~\ref{sec:distributed-ck}) makes this precise.

\paragraph{Inference-time compute and chain-of-thought reasoning.}
A growing body of work studies the power and limitations of
chain-of-thought reasoning in
transformers~\cite{feng2023towards,li2024chain,mirtaheri2025letmethink}.
Feng, Zhang, Gu, Ye, He, and Wang~\cite{feng2023towards} and Li and
Liu~\cite{li2024chain} show that CoT enables bounded-depth
transformers (TC$^0$) to solve problems requiring greater
computational depth.
Mirtaheri, Anil, Agarwal, and
Neyshabur~\cite{mirtaheri2025letmethink} prove an
exponential separation between sequential and parallel scaling for
graph connectivity.
These results concern \emph{computational} depth---the inherent
sequentiality of evaluating a function.
Our exponential separation (Section~\ref{sec:exp-separation})
establishes a distinct source of sequential advantage:
\emph{determination} depth, arising from non-commuting commitments
in relational specifications.
The separation is an offline result (the full constraint chain is
given); the decomposition of Appendix~\ref{sec:cot-appendix}
extends to the online autoregressive setting, unifying
computational depth, determination depth, and architectural
overhead into a single framework.
Complementary empirical work identifies failure modes of extended
reasoning~\cite{chen2024illusion,sui2025dontoverthink} and
studies optimal compute
allocation~\cite{snell2024scaling}; our decomposition provides a
formal account of why these phenomena arise.

\paragraph{Summary.}
Across these domains, notions of depth have appeared as proxies
for sequentiality, staging, or adaptivity.
Our contribution is to identify determination depth and
determination cost as semantic complexity measures---arising
from forward validity constraints in the online setting and
from commitment dependency in the offline setting---that are
distinct from computational complexity.
Determination depth measures the minimum sequential layers of
irrevocable commitment; determination cost measures the total
number of commitments; the exponential separation
(Theorem~\ref{thm:exp-separation}) shows that reducing one
without increasing the other is impossible.

\section*{Acknowledgments} 
Thanks to Peter Alvaro,
Tyler Hou, Sarah Morin and Christos Papadimitriou for helpful feedback on earlier drafts.

\paragraph{AI Disclosure.}
We used Claude (Anthropic), ChatGPT (OpenAI), and Gemini
(Google) to assist with exposition, structural organization,
proof review, and adversarial feedback throughout the
manuscript.
The tools materially affected the prose and presentation
in all sections, including suggestions for tone and framing of
arguments, and potential applications and open problems.
All mathematical content, definitions, theorems, and proofs
are the authors' own.
The authors verified the correctness and originality of all
content including references.

\bibliographystyle{alpha}
\bibliography{references}

\end{document}